\documentclass[11pt, letterpaper]{article}
\usepackage{amsmath, amssymb, amsthm, amsfonts, bbm, commath}
\usepackage{enumerate}
\usepackage{tikz}
\usetikzlibrary{arrows.meta,shapes}
\usepackage{authblk}
\usepackage{graphicx}
\usepackage{hyperref}
\usepackage{xcolor}
\usepackage[ruled]{algorithm2e}
\usepackage{cancel}
\usepackage{multirow}
\usepackage{calc}
\usepackage[letterpaper, margin=1in]{geometry}

\title{Exact Byzantine Consensus on Arbitrary Directed Graphs under Local Broadcast Model\footnote{This research is supported in part by the National Science Foundation awards 1409416 and 1733872, and Toyota InfoTechnology Center. Any opinions, findings, and conclusions or recommendations expressed here are those of the authors and do not necessarily reflect the views of the funding agencies or the U.S. government.}}

\author[1]{Muhammad Samir Khan}
\affil[1]{
	Department of Computer Science\protect\linebreak
	University of Illinois at Urbana-Champaign\protect\linebreak \texttt{mskhan6@illinois.edu}\protect\linebreak
}

\author[2]{Lewis Tseng}
\affil[2]{
	Department of Computer Science\protect\linebreak
	Boston College\protect\linebreak
	\texttt{lewis.tseng@bc.edu}\protect\linebreak
}

\author[3]{Nitin H. Vaidya}
\affil[3]{
	Department of Computer Science\protect\linebreak
	Georgetown University\protect\linebreak
	\texttt{nitin.vaidya@georgetown.edu}\protect\linebreak
}

\clearpage{}
\newtheorem{theorem}{Theorem}[section]
\newtheorem{lemma}[theorem]{Lemma}

\newtheorem{definition}[theorem]{Definition}

\newtheorem{observation}[theorem]{Observation}

\newenvironment{newproof}[1][Proof]{
	\begin{oldproof}[\textbf{\em #1:}]}
	{\end{oldproof}}
\renewenvironment{proof}{\begin{newproof}}{\end{newproof}}

\newcommand{\propagate}[1]{\overset{#1}{\rightsquigarrow}}
\newcommand{\notpropagate}[1]{\overset{#1}{\not \rightsquigarrow}}

\newcommand{\adjacent}[1]{\overset{}{\rightarrow}}
\newcommand{\notadjacent}[1]{\overset{}{\not \rightarrow}}

\newcommand{\inneighborhood}[3]{\Gamma_{#3}(#1 , #2)}

\newcommand{\floor}[1]{\left\lfloor {#1} \right\rfloor}

\begin{document}
	\maketitle

	\begin{abstract}
		We consider Byzantine consensus in a synchronous system where nodes are connected by a network modeled as a \emph{directed graph}, i.e., communication links between neighboring nodes are not necessarily bi-directional. The directed graph model is motivated by wireless networks wherein asymmetric communication links can occur.
In the classical point-to-point communication model, a message sent on a communication link is private between the two nodes on the link.
This allows a Byzantine faulty node to {\em equivocate}, i.e., send inconsistent information to its neighbors.
This paper considers the {\em local broadcast} model of communication, wherein transmission by a node is received identically by {\em all of its outgoing neighbors}.
This allows such neighbors to detect a faulty node's attempt to equivocate, effectively depriving the faulty nodes of the ability to send conflicting information to different neighbors.

Prior work has obtained sufficient and necessary conditions on {\em undirected graphs} to be able to achieve Byzantine consensus under the {\em local broadcast} model. In this paper, we obtain tight conditions on {\em directed graphs} to be able to achieve Byzantine consensus with binary inputs under the {\em local broadcast} model. The results obtained in the paper provide insights into the trade-off between directionality of communication and the ability to achieve consensus. 	\end{abstract}

	\section{Introduction} \label{section introduction}
		Byzantine consensus \cite{Pease:1980:RAP:322186.322188} is a classical problem in distributed computing. We consider a synchronous system consisting of $n$ nodes, each with a {\em binary} input.
The objective is for the nodes to reach consensus in the presence of up to $f$ Byzantine faulty nodes.
The nodes are connected by a communication network represented by a graph.
In the classical point-to-point communication model, a message sent by a node to one of its neighboring nodes is received only by that node. This allows a Byzantine faulty node to send inconsistent messages to its
neighbors without the inconsistency being observed by the neighbors. For instance, a faulty node $u$ may report to neighbor $v$ that its input is 0, whereas
report to neighbor $w$ that its input is 1.
Node $v$ will not hear the message sent by node $u$ to node $w$.
This ability of a faulty node to send conflicting information on different communication links is called \emph{equivocation} \cite{Chun:2007:AAM:1323293.1294280}.
The problem of Byzantine consensus in  point-to-point networks is well-studied \cite{Attiya:2004:DCF:983102, DOLEV198214, Lamport:1982:BGP:357172.357176, Lynch:1996:DA:2821576, Pease:1980:RAP:322186.322188}.
In undirected graphs, it is known that $n>3f$ and connectivity  $\geq 2f+1$ are necessary and sufficient conditions to achieve Byzantine consensus \cite{DOLEV198214}.

This paper considers the {\em local broadcast} model of communication \cite{Bhandari:2005:RBR:1073814.1073841, Koo:2004:BRN:1011767.1011807}.
Under local broadcast, a message sent by a node is received identically by all neighbors of that node.
This allows the neighbors of a faulty node to detect its attempts to equivocate, effectively depriving the faulty node of the ability to send conflicting information to its neighbors. In the example above, in the local broadcast model, if node $u$ attempts to send different input values to different neighbors, the neighbors will receive all the messages, and can detect the inconsistency.
Recent work has shown that this ability to detect equivocation reduces network requirements for Byzantine consensus in \emph{undirected} graphs.
In particular, for the local broadcast model,
recent work \cite{Khan2019ExactBC, NaqviBroadcast} has identified that the following two conditions are both necessary and sufficient for Byzantine consensus in \emph{undirected} graphs:
 network connectivity $\geq \floor{3f/2} + 1$ and minimum node degree $\geq 2f$.

In this paper, we study Byzantine consensus in \emph{directed} graphs under the local broadcast model.
The directed graph model is motivated by wireless networks wherein asymmetric links may occur. Thus, the communication links between neighboring nodes are not necessarily bi-directional. Under local broadcast, when some node $u$ transmits a message, it is received identically by all of $u$'s {\em outgoing} neighbors (i.e., by nodes to whom there are outgoing links from node $u$).  The results obtained in the paper provide insights into the trade-off between directionality of communication and the ability to achieve consensus.

This paper makes two main contributions:
\begin{enumerate}
	\item \textbf{Necessity:}
	In Section \ref{section necessity}, we identify a necessary condition for directed graphs to solve Byzantine consensus.
	The proof is via a state machine based approach  \cite{Attiya:2004:DCF:983102, DOLEV198214, Fischer1986}.
	However, the communication by faulty nodes must follow the local broadcast model which restricts their behavior.
	The directed nature of the graph also adds to the difficulty.
	We handle this complexity via non-trivial arguments to show the desired result.

	\item \textbf{Sufficiency:}
	In Section \ref{section sufficiency}, we constructively show that the necessary condition is also sufficient by presenting a Byzantine consensus algorithm along with its proof of correctness.
	The key challenge for directed graphs is that communication may only exist in one direction between some pairs of nodes.
	Hence, it is not straightforward to adapt the prior algorithm \cite{Khan2019ExactBC} on undirected graphs.
	More specifically, in undirected graphs, each node has some path(s) to each of the other nodes\footnote{In undirected graphs, consensus is not possible if graph is not connected.}.
	Prior algorithm by Khan et. al. \cite{Khan2019ExactBC} utilizes this property to solve Byzantine consensus in undirected graphs under local broadcast.
	In the directed case, this property is not guaranteed due to the directionality of communication.
	We prove some non-trivial properties (Section 6.2) implied by the tight condition identified in this paper.
	This allows us to achieve consensus in a unique “source component” and then propagate that decision to the rest of the graph.
\end{enumerate}

The rest of the paper is organized as follows.
We discuss related work in Section \ref{section related work}.
Section \ref{section notation} formalizes the setting and introduces notation.
Our main result is presented in Section \ref{section main results}.
Necessity of the identified tight condition is shown in Section \ref{section necessity}
while
sufficiency is shown in Section \ref{section sufficiency}.
We summarize in Section \ref{section conclusion}.
 
	\section{Related Work} \label{section related work}
		Byzantine consensus is a well-studied problem \cite{Attiya:2004:DCF:983102, DOLEV198214, Lamport:1982:BGP:357172.357176, Lynch:1996:DA:2821576, Pease:1980:RAP:322186.322188} with tight conditions known for complete graphs \cite{Pease:1980:RAP:322186.322188}, undirected graphs \cite{DOLEV198214}, and  directed graphs \cite{LewisByzantineDirected} under the point-to-point communication model.
For brevity, we focus here on related work that restricts equivocation by faulty nodes.

Rabin and Ben-Or \cite{Rabin:1989:VSS:73007.73014} considered complete graphs with global broadcast under synchronous communication, while 
Clement et. al. \cite{Clement:2012:PN:2332432.2332490} looked at non-equivocation in complete graphs under asynchronous communication.
Amitanand et. al. \cite{Amitanand:2003:DCP:872035.872065} restricted equivocation by faulty nodes by partitioning, for each faulty node $w$, the remaining graph such that a message sent by $w$ to any node is received identically by all nodes in the corresponding partition.
However, the underlying graph in \cite{Amitanand:2003:DCP:872035.872065} is complete while we consider arbitrary directed graphs.
Several works \cite{Fitzi:2000:PCG:335305.335363, Jaffe:2012:PEB:2332432.2332491, Ravikant10.1007/978-3-540-30186-8_32} have used \emph{undirected} hypergraphs to model partial broadcast for the Byzantine consensus problem.
In this model, a message sent on an hyperedge is received identically by all nodes in the hyperedge.
The closest work to this paper is by Khan et. al. \cite{Khan2019ExactBC, NaqviBroadcast}, who obtained that minimum node degree $\geq 2f$ and network connectivity $\geq \floor{3f/2} + 1$ are tight conditions for Byzantine consensus in \emph{undirected} graphs under the local broadcast model.
Here, we consider arbitrary \emph{directed} graphs under the local broadcast model.

Restricted equivocation has also been used to study related problems.
\cite{Considine2005, Franklin2000, FranklinHypergraphsPrivacy2004, Wang2001} looked at reliability and privacy on partial broadcast networks.
\cite{Bhandari:2005:RBR:1073814.1073841, Koo:2004:BRN:1011767.1011807, Koo:2006:RBR:1146381.1146420}
have investigated the Byzantine \emph{broadcast} problem under local broadcast on both undirected and directed graphs.
In Byzantine broadcast, the goal is for a single source to transmit a binary value reliably throughout the network.
We consider the Byzantine \emph{consensus} problem, where the goal is for all nodes to agree on a common value.

Another line of work investigates iterative algorithms for \emph{approximate} Byzantine consensus.
In this problem, each node starts with a real value (or a vector of real values) and maintains a state variable.
The updates are ``memory less'', i.e., the update rules allow a state update in each round to depend only on the current state and the state values received from neighbors.
This problem has been investigated under the classical point-to-point communication model on directed graphs by Tseng and Vaidya \cite{TsengIterativeICDCN} and Vaidya et. al. \cite{VaidyaICDCN_Vector, Vaidya:2012:IAB:2332432.2332505}, under partial broadcast modeled via directed hypergraphs by Li et. al. \cite{Li7516067}, and under the local broadcast model on directed graphs by LeBlanc et. al. \cite{LeBlancIterative} as well as by Zhang and Sundaram \cite{ZhangIterative}.
The network conditions are different than the ones presented in this paper, since the algorithm structure is restricted (as summarized above) in these prior works. 
	\section{System Model and Notation} \label{section notation}
		We consider a synchronous system consisting of $n$ nodes. The communication network connecting the nodes is represented by a \emph{directed} graph $G = (V, E)$, where $\abs{V} = n$.
Each of the $n$ nodes is represented by a vertex $u \in V$.
We interchangeably use the terms \emph{node} and \emph{vertex}.
Every node in the graph knows the communication graph $G$.
Each directed edge $(u, v) \in E$ represents a FIFO link from $u$ to $v$.
When a message $m$ sent by node $u$ on edge $(u,v)$ is received by node $v$, node $v$ knows that the message $m$ was sent by node $v$.
This assumption is implicit in the previous related work as well.
We assume the \emph{local broadcast} model of communication wherein a message sent by any node $u$ is received identically and correctly by each node $v$ such that $(u, v) \in E$.

A \emph{Byzantine} faulty node may exhibit arbitrary behavior,
however, its communication is still governed by the local broadcast model.
We consider the \emph{Byzantine consensus problem}.
Each node starts with a binary input and must output a binary value.
There are at most $f$ Byzantine faulty nodes in the system, where $0 < f < n$\footnote{The case with $f = 0$ is trivial and the case with $n = f$ is not of interest.}.
The output at each node must satisfy the following conditions.

\begin{enumerate}
	\item \textbf{Agreement:}
	All non-faulty nodes must output the same value.

	\item \textbf{Validity:}
	The output of each non-faulty node must be an input of some non-faulty node.

	\item \textbf{Termination:}
	All non-faulty nodes must decide on their output in finite time.
\end{enumerate}

\begin{enumerate}[\null]
	\item \textbf{Neighborhood:}
	If $(u, v)\in E$, then $u$ is an \emph{in-neighbor} of $v$ and $v$ is an \emph{out-neighbor} of $u$.
	The \emph{in-neighborhood} of a node $v$ is the set of all in-neighbors of $v$, i.e., $\set{ u \mid (u, v) \in E }$.
	Similarly, the \emph{out-neighborhood} of a node $v$ is the set of all out-neighbors of $v$, i.e., $\set{ u \mid (v, u) \in E }$.
	In graph $G$, for node sets $A$ and $B$, we define in-neighborhood of set $B$ in set $A$, denoted $\inneighborhood{A}{B}{G}$,
	as the set of in-neighbors of nodes in $B$ that are in set $A$. That is, $\inneighborhood{A}{B}{G} = \set{ u \in A \mid \exists v \in B \: \text{s.t.} \: (u, v) \in E(G)}$.
	Note that $E(G)$ denotes the set of edges in graph $G$. We will use the above definition for different graphs, hence the subscript $G$ above is important. We may drop the subscript $G$ when it is clear from the context.

	We will say that $A \adjacent{F}_G B$ if $\abs[0]{ \inneighborhood{A}{B}{G} } > f$. Here as well, we may drop the subscript $G$ when it is clear from the context.

	\item \textbf{Paths in graph $G$:}
	A path is a sequence of nodes such that if $u$ precedes $v$ in the sequence, then $u$ is an in-neighbor of $v$ in $G$ (i.e., $(u, v)$ is an edge).
	\begin{itemize}
		\item
		For two nodes $u$ and $v$, a $uv$-path $P_{uv}$ is a path from $u$ to $v$.
		$u$ is called the \emph{source} and $v$ the \emph{terminal} of $P_{uv}$.
		Any other node in the path is called an \emph{internal} node of $P_{uv}$.
		Two $uv$-paths are \emph{node-disjoint} if they do not share a common internal node.

		\item
		For a set $U \subsetneq V$ and a node $v \not \in U$, a $Uv$-path is a $uv$-path for some node $u \in U$.
		All $Uv$-paths have $v$ as the terminal.
		Two $Uv$-paths are \emph{node-disjoint} if they do not have any nodes in common except terminal node $v$.
		In particular, two node-disjoint $Uv$-paths have different source nodes.
	\end{itemize}
	A path is said to \emph{\underline{exclude}} a set of nodes $X \subset V$ if no internal node of the path belongs to $X$; however, its source and terminal nodes may potentially belong to $X$.
	A path is said to be \emph{\underline{fault-free}} if none of its internal nodes are faulty.
	In other words, a path is fault-free if it excludes the set of faulty nodes.
	Note that a fault-free path may have a faulty node as either source or terminal.

	We use the notation $A \propagate{X}_G B$ if, for every node $u \in B$, there exist at least $f+1$ node-disjoint $Au$-paths in $G$ that exclude $X$, i.e., there exist $f+1$ node-disjoint $Au$-paths that have only $u$ in common and none of them contain any internal node from the set $X$.
	We may omit the subscript $G$ when it is clear from the context.
\end{enumerate}

With a slight abuse of terminology, we allow a \emph{partition} of a set to have empty parts.
That is, $(Z_1, \dots, Z_k)$ is a partition of a set $Y$ if $\cup_{i=1}^{k} Z_i = Y$ and $Z_i \cap Z_j = \emptyset$ for all $i \ne j$, but some $Z_i$'s can be possibly empty.

For a set of nodes $U \subsetneq V$,
\begin{itemize}
	\item $G[U]$ is the subgraph induced by the nodes in $U$.
	\item $G_{-U}$ is the graph obtained from $G$ by removing all edges $(v, u)$ such that $u \in U$, i.e., by removing all incoming edges to $U$.
	Observe that if $P$ is a path in $G_{-U}$, then $P$ is a path in $G$ that excludes $U$ and terminates in $V - U$.
	Conversely, if $P$ is a path in $G$ that excludes $U$ and terminates in $V - U$, then $P$ is a path in $G_{-U}$.
\end{itemize}

A directed graph $G$ is \emph{strongly connected} if for each pair of nodes $u, v$, there is both a $uv$-path and a $vu$-path in $G$.
A \emph{directed graph decomposition} of $G$ is a partition of $G$ into non-empty parts $H_1, \dots, H_k$, where $k > 0$, and each $H_i$ is a maximal strongly connected subgraph of $G$ -- each $H_i$ is assumed to be maximal in the sense that adding any nodes to $H_i$ will destroy its strong connectivity.
Let $\mathcal{H}$ be the graph obtained from the decomposition by contracting each $H_i$ into a node $c_i$, so that there is an edge $(c_i, c_j)$ in $\mathcal{H}$ if there is an edge from a node in $H_i$ to a node in $H_j$ in $G$.
Then, graph $\mathcal{H}$ is acyclic.
If a node $c_i$ has no in-neighbors, then $H_i$ is called a \emph{source component} of the decomposition.
Note that, since $\mathcal{H}$ is acyclic, there is always at least one source component of a directed graph decomposition.
 
	\section{Main Results} \label{section main results}
		The main result of this paper is a tight network condition for consensus in directed graphs under the local broadcast model.
The following definition presents the condition and the accompanying theorem states the result.

\begin{definition} \label{definition condition SC}
	A directed graph $G$ satisfies \underline{\emph{condition SC with parameter $F$}} if for every partition $(A, B)$ of $V$, where both $A - F$ and $B - F$ are non-empty, we have that either $A \propagate{F} B - F$ or $B \propagate{F} A - F$.
	We say that $G$ satisfies \underline{\emph{condition SC}}, if $G$ satisfies condition SC with parameter $F$ for every set $F \subseteq V$ of cardinality at most $f$.
\end{definition}

\begin{theorem} \label{theorem tight}
	Under the local broadcast model, Byzantine consensus tolerating at most $f$ Byzantine faulty nodes is achievable on a directed graph $G$ if and only if $G$ satisfies condition SC.
\end{theorem}
\begin{proof}
	The proof follows from Theorems \ref{theorem equivalence}, \ref{theorem necessity}, and \ref{theorem sufficiency} presented later.
\end{proof}

Intuitively, the above condition requires that at least one of the two partitions A and B should have the ability to ``propagate'' its state to the other partition reliably.
For the point-to-point communication model, Tseng and Vaidya \cite{LewisByzantineDirected} obtained an analogous network condition, which is that, for every partition $(A, B)$ of $V$ and a faulty set $F$, where both $A - F$ and $B - F$ are non-empty, either $A - F \propagate{F} B - F$ or $B - F \propagate{F} A - F$.
In the point-to-point communication model a faulty node can equivocate.
Thus, the condition in \cite{LewisByzantineDirected} does not allow nodes in set $F$ to be source nodes\footnote{Recall that, in a $uv$-path, $u$ is the source node and $v$ is the terminal.}, and requires $A-F$ or $B-F$ to propagate its state to the non-faulty
nodes in the other partition. On the other hand,
as discussed earlier, local broadcast effectively removes a faulty node's ability to equivocate.
Therefore, the condition in Definition \ref{definition condition SC} allows a node in set $F$ to be a source node in the propagation paths, but does not allow nodes in $F$ to be internal nodes on such paths.

Even though the conditions for point-to-point communication \cite{LewisByzantineDirected} and local broadcast seem similar, the algorithm in \cite{LewisByzantineDirected} is not immediately adaptable to the local broadcast model.
We discuss this at the end of Section \ref{section proposed algorithm}.

We prove necessity of condition SC via a state machine based approach  \cite{Attiya:2004:DCF:983102, DOLEV198214, Fischer1986} similar to the proofs of necessity in \cite{Khan2019ExactBC, LewisByzantineDirected}.
However, care must be taken to ensure that we do not break the local broadcast property.
The formal proof is given in Appendix \ref{section proof necessity} -- since it is somewhat difficult.
In Section \ref{section necessity}, we provide an intuitive sketch of the proof.
The sufficiency is proved constructively.
In Section \ref{section sufficiency}, we present an algorithm to achieve consensus when the communication graph satisfies condition SC, accompanied by a proof of correctness.

	\section{Necessity} \label{section necessity}
		In this section, we show that condition SC (Definition \ref{definition condition SC}) is necessary for consensus.
We first present another property, condition NC.
This condition is equivalent to condition SC, as stated in the next theorem, and we will use it to prove necessity in Theorem \ref{theorem necessity}.
Recall that we use $A \adjacent{F}_G B$ to denote $\abs[0]{ \inneighborhood{A}{B}{G} } > f$ (subscript $G$ is dropped when clear from context).

\begin{definition} \label{definition condition NC}
	A directed graph $G$ satisfies \underline{\emph{condition NC with parameter $F$}} if for every partition $(L, C, R)$ of $V$, where both $L - F$ and $R - F$ are non-empty, we have that either ${R \cup C} \adjacent{F}_G L - F$ or ${L \cup C} \adjacent{F}_G R - F$.
	We say that $G$ satisfies \underline{\emph{condition NC}}, if $G$ satisfies condition NC with parameter $F$ for every set $F \subseteq V$ of cardinality at most $f$.
\end{definition}

\begin{theorem} \label{theorem equivalence}
	A directed graph $G$ satisfies condition NC if and only if $G$ satisfies condition SC.
\end{theorem}
\begin{proof}
	From Lemma \ref{lemma simple equivalence} in Appendix \ref{section proof equivalence}
\end{proof}

The following theorem states that condition NC is necessary for consensus.
Since condition NC and condition SC are equivalent, as a corollary we get the necessity part of Theorem \ref{theorem tight}.

\begin{theorem} \label{theorem necessity}
	If there exists a Byzantine consensus algorithm under the local broadcast model on a directed graph $G$ tolerating at most $f$ Byzantine faulty nodes, then $G$ satisfies condition NC.
\end{theorem}

The proof uses a state machine based approach \cite{Attiya:2004:DCF:983102, DOLEV198214, Fischer1986}, along the lines of the proofs of necessity in \cite{Khan2019ExactBC, LewisByzantineDirected}.
However, care must be taken to ensure that all the communication is via local broadcast, which makes it somewhat difficult.
The formal proof of this theorem is given in Appendix \ref{section proof necessity}.
Here we give a sketch of the proof.

Suppose for the sake of contradiction that there exists an algorithm that
solves Byzantine consensus under the local broadcast model on a graph $G$
which does not satisfy condition NC.
Then there exists a set $F$ of cardinality at most $f$ and a partition
$\del{ L, C, R }$ of $G$, where both $L - F$ and $R - F$ are non-empty, such that $R \cup C \notadjacent{F} L - F$ and $L \cup C \notadjacent{F} R - F$.
We create three executions $E_1$, $E_2$, and $E_3$ using the algorithm as follows.

\begin{enumerate}[$E_1$:]
	\item
	$\inneighborhood{R \cup C}{L - F}{}$, the in-neighborhood of $L - F$ in $R \cup C$, is the faulty set.
	We partition the faulty set into two parts: the in-neighborhood of $L - F$ in $R - F$ and $C$, $\inneighborhood{(R-F) \cup C}{L - F}{}$, and the in-neighborhood of $L - F$ in $R \cap F$, $\inneighborhood{R \cap F}{L - F}{}$.
	Both these sets have different behavior.
	In each round, a faulty node in $\inneighborhood{(R-F) \cup C}{L - F}{}$ broadcasts the same messages as the corresponding non-faulty node in $E_3$, while a faulty node in $\inneighborhood{R \cap F}{L - F}{}$ broadcasts the same messages as the corresponding non-faulty node in $E_2$.
	All non-faulty nodes have input $0$.
	So by validity, all non-faulty nodes decide on output $0$ in finite time.

	\item
	$\inneighborhood{L \cup C}{R - F}{}$, the in-neighborhood of $R - F$ in $L \cup C$, is the faulty set.
	We partition the faulty set into two parts: the in-neighborhood of $R - F$ in $L - F$ and $C$, $\inneighborhood{(L-F) \cup C}{R - F}{}$, and the in-neighborhood of $R - F$ in $L \cap F$, $\inneighborhood{L \cap F}{R - F}{}$.
	Both these sets have different behavior.
	In each round, a faulty node in $\inneighborhood{(L-F) \cup C}{R - F}{}$ broadcasts the same messages as the corresponding non-faulty node in $E_3$, while a faulty node in $\inneighborhood{L \cap F}{R - F}{}$ broadcasts the same messages as the corresponding non-faulty node in $E_1$.
	All non-faulty nodes have input $1$.
	So by validity, all non-faulty nodes decide on output $1$ in finite time.

	\item
	$F \cap (L \cup R)$ is the faulty set.
	We partition the faulty set into two parts: $F \cap L$ and $F \cap R$.
	Both these sets have different behavior.
	In each round, a faulty node in $F \cap L$ broadcasts the same messages as the corresponding non-faulty node in $E_1$, while a faulty node in $F \cap R$ broadcasts the same messages as the corresponding non-faulty node in $E_2$.
	All non-faulty nodes in $L$ have input $0$ and all non-faulty nodes in $R \cup C$ have input $1$.
	The output of the non-faulty nodes in this execution will be described later.
\end{enumerate}

\begin{figure}[tb]
	\centering
	\begin{tikzpicture}
    	\node[draw, circle, minimum size=4cm, label={left:$L$}] at (0, 0) (L) {};
    	\node[draw, circle, minimum size=4cm, label={right:$R$}] at (5, 0) (R) {};
    	\node[draw, ellipse, minimum width=4cm, minimum height=3cm, red, label={center:$F$}] at (2.5, 0) (F) {};
    	\node[draw, ellipse, minimum width=2cm, minimum height=1cm, blue, label={below right:$\inneighborhood{R}{L}{}$}, rotate=45] at (4, 1) (NL) {};
    	\node[draw, ellipse, minimum width=2cm, minimum height=1cm, green, label={above left:$\inneighborhood{R}{L}{}$}, rotate=45] at (1, -1) (NR) {};

	    \draw[-{Latex[width=3mm,length=3mm]}, out=90, in=60] (NL) to (L);
	    \draw[-{Latex[width=3mm,length=3mm]}, out=270, in=240] (NR) to (R);
	\end{tikzpicture}
	\caption{Necessity proof: the simple case when $C = \emptyset$. The nodes in blue, green, and red depict the three faulty sets for executions $E_1$, $E_2$, and $E_3$ respectively.}
	\label{figure necessity sketch}
\end{figure}
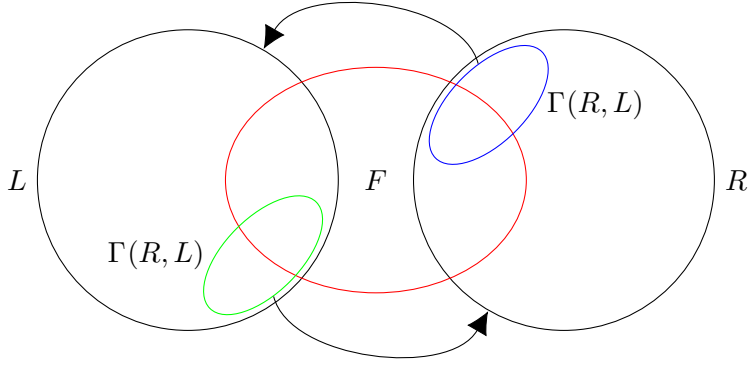

The formal proof in Appendix \ref{section proof necessity} makes the above description of the three executions precise.
For the simple case when $C = \emptyset$, Figure \ref{figure necessity sketch} depicts the faulty nodes in the three executions.

To see the output of non-faulty nodes in $E_3$, note that non-faulty nodes in $L - F$ receive the same messages in each round, from their in-neighbors, as the corresponding nodes in $E_1$.
They also have the same input, $0$, so they decide on the same output in both the executions, i.e., $0$.
Similarly, the non-faulty nodes in $R - F$ receive the same messages in each round, from their in-neighbors, as the corresponding nodes in $E_2$.
They also have the same input, $1$, so they decide on the same output in both the executions, i.e., $1$.
Since both $L - F$ and $R - F$ are non-empty, this violates agreement, a contradiction.
 
	\section{Sufficiency} \label{section sufficiency}
		Next, we constructively prove the sufficiency portion of Theorem \ref{theorem tight}.
Together with the necessity result shown in Theorem \ref{theorem necessity}, we have that this result is tight.
We present a Byzantine consensus algorithm in Section \ref{section proposed algorithm}.
The algorithm utilizes some non-trivial graph properties implied by condition SC.
We show these in Section \ref{section graph properties}.
In Section \ref{section proof of correctness}, we give a proof of correctness of the algorithm, assuming that the graph $G$ satisfies condition SC. 		\subsection{Proposed Algorithm} \label{section proposed algorithm}
			Algorithm \ref{algorithm consensus} presents pseudocode for the proposed algorithm.
Each node $v\in V$ maintains a local state variable named $\gamma_v$.
At the beginning of the algorithm, this is initialized to equal node $v$'s binary input.
$\gamma_v$ is modified during the execution of the algorithm.
The output of each node $v$ is the value of its state variable $\gamma_v$ at the end of the algorithm.
The algorithm execution is viewed as being divided into \emph{phases}, each
phase consisting of one iteration of the {\tt for} loop in the pseudocode.
Each phase has an associated distinct subset $F \subseteq V$ of cardinality at most $f$.

The algorithm draws inspiration from the strategy used in the algorithms in \cite{Khan2019ExactBC} and \cite{LewisByzantineDirected}.
However, the details are significantly different, which we discuss at the end of this section.
In particular, each phase (i.e., each iteration of the {\tt for} loop) in the algorithm considers a candidate faulty set $F$, and the nodes attempt to reach consensus by the end of that phase assuming that $F$ is indeed the set of all the faulty nodes.
Let $F^*$ denote the set of faulty nodes that are actually faulty in a given execution of the algorithm.
Then each node updates its state variable in such a manner so that
\begin{enumerate}[(i)]
	\item when $F = F^*$, the state variable at all non-faulty nodes is identical at the end of this phase, i.e., all non-faulty nodes reach consensus in this phase (Lemma \ref{lemma agreement}), and
	\item when $F \ne F^*$, the value of the state variable at a non-faulty node at the end of the phase equals the state of some non-faulty node at the start of the phase (Lemma \ref{lemma validity}).
\end{enumerate}
The first objective ensures that nodes reach agreement.
The second inductively implies that this agreement, once achieved, is not lost, as well as that nodes decide on an input of some non-faulty node, i.e., validity.
Termination follows from the fact that there are only finite number of executions.

We now discuss the steps performed by each node $v$ in a given phase.
Some of the steps of the algorithm are based on those in \cite{Khan2019ExactBC} and are explained here again for completeness.

\begin{algorithm}[t]
\SetAlgoLined
\SetKwFor{For}{For}{do}{end}
Each node $v$ has a binary input value in $\set{ 0, 1 }$ and maintains a binary state $\gamma_v \in \set{ 0, 1 }$.\\
{\em Initialization:} $\gamma_v$ := input value of node $v$.\\
\For{each $F \subseteq V$ such that $\abs{F} \le f$}
{
	\smallskip
	\setlength{\linewidth}{\columnwidth-\widthof{Step (a):}}
	\begin{enumerate}[Step (a):]
		\item
		Perform directed graph decomposition on $G - F$.
		Let $S$ be the unique source component (Lemma \ref{lemma source component unique}).
		
		\item
		If $v \in S \cup \inneighborhood{F}{S}{}$, then flood value $\gamma_v$
		(the steps taken to achieve flooding are described in Appendix \ref{section flooding}).
		
		\item
		If $v \in S$, for each node $u \in S \cup \inneighborhood{F}{S}{}$, identify a single $uv$-path $P_{uv}$ that excludes $F$.
		Let,
		\begin{flalign*}
		Z_v &:= \set{ u \in S \cup \inneighborhood{F}{S}{} \mid \text{$v$ received value $0$ from $u$ along $P_{uv}$ in step (b)} },&\\
		N_v &:= S \cup \inneighborhood{F}{S}{} - Z_v.&
		\end{flalign*}
		
		\item
		If both $Z_v - F$ and $N_v - F$ are non-empty, then\\
		\smallskip
		\hspace*{0.3in} If $Z_v \propagate{F} N_v - F$,\\
		\hspace*{0.6in} then set $A_v := Z_v$ and $B_v := N_v - F$,\\
		\hspace*{0.6in} else set $A_v := N_v$ and $B_v := Z_v - F$.\\
		\smallskip
		\hspace*{0.3in} If $v \in B_v$ and $v$ received value $\delta\in\{0,1\}$, in step (b), identically\\
		\hspace*{0.3in} along any $f+1$ node-disjoint $A_v v$-paths that exclude $F$, then set\\
		\hspace*{0.3in} $\gamma_v := \delta$.
		
		\item
		If $v \in S$, then flood value $\gamma_v$.
		
		\item
		If $v \in V - S - F$ and $v$ received value $\delta \in \set{0, 1}$, in step (e), identically along any $f+1$ node-disjoint $Sv$-paths that exclude $F$, then set $\gamma_v := \delta$.
	\end{enumerate}
}
Output $\gamma_v$.

\caption{
	Proposed algorithm for Byzantine consensus under the local broadcast model in directed graphs: Steps performed by node $v$ are shown here.
}
\label{algorithm consensus}
\end{algorithm} 
\begin{enumerate}[\null]
	\setlength{\itemsep}{\smallskipamount}
	\item \underline{\em Graph decomposition:}
	In step (a) of a given phase,
	each node performs a directed graph decomposition on $G - F$.
	Recall that we assume that each node knows the topology of graph $G$, so each node can perform this step locally.
	Since $G$ satisfies condition SC with parameter $F$, it turns out that $G - F$ has a unique source component $S$ (Lemma \ref{lemma source component unique}).
	The rest of the steps in the phase are aimed at nodes in $S$ attempting to agree on some common value, and then propagating that value to the rest of the graph.

	\item \underline{\em Flooding:}
	In Step (b), nodes in $S$ and their in-neighbors in $F$, $\inneighborhood{F}{S}{}$, flood the value of their $\gamma$ state variables. The ``flooding'' procedure used here
	is analogous to that in \cite{Khan2019ExactBC} for undirected graphs. Without much modification, it can be adapted for directed graphs. This procedure is presented in Appendix \ref{section flooding} for completeness.

	Consider a node $u \in S \cup \inneighborhood{F}{S}{}$.
	In the flooding procedure, $u$ attempts to transmit its state variable to every node $v$ such that there exists a $uv$-path in $G$.
	At the end of the flooding procedure, for each such node $v$ and a $uv$-path $P_{uv}$, $v$ will have received a binary value $b$ along $P_{uv}$.
	If $P_{uv}$ is fault-free, then $b = \gamma_u$, the state variable of $u$.
	However, if $P_{uv}$ is not fault-free, then an intermediary faulty node may tamper with the messages, so it is possible that $b \ne \gamma_u$.

	\item \underline{\em Consensus in the source component $S$:}
	Next, using steps (c) and (d), the nodes in the unique source component try to reach consensus. Based on the values flooded by each node, set $S \cup \inneighborhood{F}{S}{}$ can be partitioned into nodes that flooded $0$, namely set $Z$, and nodes that flooded $1$, namely set $N$.
	In step (c), node $v$ attempts to estimate the sets $Z$ and $N$ using the values received on paths excluding $F$, i.e., ignoring paths that have nodes from $F$ as intermediaries.
	When $F \ne F^*$, faulty nodes may tamper the messages and $v$ may incorrectly categorize some nodes.
	However, when $F = F^*$, all non-faulty nodes correctly determine $Z$ and $N$.

	In particular, in step (c), each node $v \in S$ partitions $S \cup \inneighborhood{F}{S}{}$ as follows.
	Recall that a path is said to exclude set $F$ if none of its {\em internal} nodes are in $F$.
	For each $u \in S \cup \inneighborhood{F}{S}{}$, node $v$ chooses an arbitrary $uv$-path $P_{uv}$ that excludes $F$.
	It can be shown (Lemma \ref{lemma paths}) that such a path always exists.
	For the purpose of step (c), node $v$ is deemed to have received its own $\gamma_v$ value along path $P_{vv}$ (containing only node $v$).
	In step (c), as shown in the pseudo-code, node $v$ partitions $S \cup \inneighborhood{F}{S}{}$ into sets $Z_v$ and $N_v$, its estimates of sets $Z$ and $N$, based
	on values received along the above paths.

	Step (d) specifies the rules for updating $\gamma_v$ value.
	$\gamma_v$ is not necessarily updated in each phase.
	If $F = F^*$, then all nodes in $S$ will have the same $\gamma_v$ value after this step (Lemma \ref{lemma agreement}). That is, in the phase in which $F=F^*$,
	the nodes in $S$ achieve consensus in step (d).

	\item \underline{\em Propagating decision to rest of the graph:}
	Since we iterate over all possible faulty sets, in one phase of the algorithm, $F$ is correctly chosen to be exactly $F^*$, the set of actual faulty nodes.
	In this phase, nodes in $S$ will reach consensus by step (d).
	In steps (e) and (f), nodes in $S$ propagate their state to the rest of the graph.
	In particular, in step (e), nodes in $S$ flood their $\gamma_v$ value, as in step (b), while step (f) specifies the rules for nodes in $V - S - F$ to update their $\gamma_v$ state variable.

	\item \underline{\em Output:}
	After all the phases (i.e., all iterations of the {\tt For} loop) are completed,
	the value of $\gamma_v$ is chosen as the output of node $v$.
\end{enumerate}

As mentioned earlier, the proposed algorithm uses the same strategy as in \cite{Khan2019ExactBC} and \cite{LewisByzantineDirected} by iterating over all possible faulty sets.
However, the steps performed in each iteration are significantly different than both.

\begin{itemize}
	\item
	The algorithm for directed graphs under point-to-point communication in \cite{LewisByzantineDirected} is not immediately adaptable to the local broadcast setting.
	One key challenge is that it requires nodes to send messages in multiple rounds in a single iteration of the main for loop.
	While non-equivocation provided by local broadcast prevents a faulty node from sending conflicting messages in one round, it does not directly stop a faulty node from lying inconsistently across rounds, even to the same neighbor.
	In our algorithm, when $F = F^*$, the faulty nodes are only allowed to flood their states once in the entire iteration, preventing them from lying inconsistently across rounds in this iteration.

	\item
	On the other hand, the algorithm for undirected graphs under local broadcast \cite{Khan2019ExactBC} does indeed require each faulty node to send a single message for each iteration of the main for loop.
	However, in directed graphs, in contrast with the undirected setting, communication may only exist in one direction between some pairs of nodes.
	Our algorithm first achieves consensus in a unique source component and then propagates this to the rest of the graph.
	The existence of a \emph{unique} source component or its ability to propagate its state to the rest of the graph is not immediately obvious.
	In the next section, we show that these and other non-trivial properties are guaranteed by condition SC, which our algorithm relies to solve consensus.
\end{itemize} 		\subsection{Graph Properties} \label{section graph properties}
			Algorithm \ref{algorithm consensus} relies on some non-trivial properties of graphs that satisfy condition SC.
In this section, we prove these properties followed by the proof of correctness of the algorithm in Section \ref{section proof of correctness}.
We first show that there is indeed a unique source component in the directed graph decomposition of $G - F$, performed in step (a) of each phase.
Suppose, for the sake of contradiction, that there are two source components $S_1$ and $S_2$ of the decomposition.
Note that, by construction, there are no edges into $S_1$ and $S_2$, except from $F$.
By appropriately selecting sets $A$ and $B$, the only paths into $A$ (resp. $B$) are via $F$ and so are limited to at most $f$, violating condition SC with parameter $F$, a contradiction.

\begin{lemma} \label{lemma source component unique}
	For any choice of set $F$ in the algorithm the directed graph decomposition of $G - F$ has a unique source component.
\end{lemma}

\begin{proof}
	Fix an arbitrary set $F$.
	Suppose for the sake of contradiction that $G - F$ has two source components $S_1$ and $S_2$.
	Let $C := V - S_1 - S_2$ be the rest of the nodes.
	We create a partition $(A, B)$ that violates the requirements of condition SC with parameter $F$ (Definition \ref{definition condition SC}).
	Let $A := S_1$ and $B := S_2 \cup C = V - S_1$.
	First note that both $A - F = A = S_1$ and $B - F \supseteq S_2$ are non-empty.
	
	Now, by construction, we have that $S_2$ has no incoming edges except from $F$.
	Therefore, $F$ is a cut set that separates $S_2$ from $S_1$, i.e., there are no paths from a node in $S_1$ to a node in $S_2$ in $G - F$.
	By Menger's Theorem, there are at most $\abs{F} \le f$ node-disjoint paths from $S_1 = A$ to any node in $S_2 \subseteq B - F$.
	Therefore, for any node $v \in S_2 \subseteq B - F$, there are at most $f$ node-disjoint $A v$-paths that exclude $F$, and so $A \notpropagate{F}{} B - F$.
	Similarly, by construction, we have that $S_1$ has no incoming edges except from $F$.
	Therefore, $F$ is a cut set that separates $S_1$ from $S_2 \cup C - F$.
	By Menger's Theorem, there are at most $\abs{F} \le f$ node-disjoint paths from $S_2 \cup C = B$ to any node in $S_1 = A - F$.
	Therefore, for any node $v \in S_1 = A - F$, there are at most $f$ node-disjoint $B v$-paths that exclude $F$, and so $B \notpropagate{F}{} A - F$.
	This violates condition SC, a contradiction.
\end{proof}

Next, this unique source component $S$, along with its in-neighborhood in $F$, satisfies condition SC with parameter $F$.
We follow the same approach as in proof of Lemma \ref{lemma source component unique}, assuming for the sake of contradiction that $G[ S \cup \inneighborhood{F}{S}{} ]$ does not satisfy condition SC with parameter $F$, and showing that this implies that $G$ does not satisfy condition SC with parameter $F$.

\begin{lemma} \label{lemma source component satisfies condition SC}
	For any choice of set $F$ in the algorithm, let $S$ be the unique source component of $G - F$.
	Then $G[S \cup \inneighborhood{F}{S}{}]$ satisfies condition SC with parameter $F$.
\end{lemma}

\begin{proof}
	Suppose for the sake of contradiction that $H := G[S \cup \inneighborhood{F}{S}{}]$ does not satisfy condition SC with parameter $F$.
	So there exists a partition $(A, B)$ of $S \cup \inneighborhood{F}{S}{}$, such that $A - F$ and $B - F$ are non-empty, and $A \notpropagate{F}_H B - F$ and $B \notpropagate{F}_H A - F$.
	Let the rest of the nodes in $G$ be denoted by $C := V - S \cup \inneighborhood{F}{S}{}$.
	Let $A' := A \cup C$.
	Then $(A', B)$ is a partition of $V$, with $A' - F$ and $B - F$ both non-empty.
	
	We first show that there is no edge from $C$ to $S$ in $G$.
	Observe that $C$ is disjoint from $S \cup \inneighborhood{F}{S}{}$.
	If a node $u \in C \cap F$ has an edge to a node in $S$, then $u \in \inneighborhood{F}{S}{}$, a contradiction since $C \cap \inneighborhood{F}{S}{} = \emptyset$.
	On the other hand, if $u \in C - F$ has an edge to a node in $S$, then the edge exists in $G - F$, which is a contradiction since $S$ is a source component in $G - F$.
	
	Now, since $A \notpropagate{F}_H B - F$, we have that for some node $v \in B - F$ there are at most $f$ node-disjoint $Av$-paths in $H_{-F}$\footnote{Recall from Section \ref{section notation} that $G_{-U}$ is the graph obtained from $G$ by removing all incoming edges to $U$ so that if $P$ is a path in $G_{-U}$, then $P$ is a path in $G$ that excludes $U$ and terminates in $V - U$, and if $P$ is a path in $G$ that excludes $U$ and terminates in $V - U$, then $P$ is a path in $G_{-U}$.}.
	By Menger's Theorem, there exists a cut set $X \not \ni v$ of cardinality at most $f$ that separates $v$ from $A - X$ in $H_{-F}$.
	Since there is no edge from $C$ to $S$ in $G$, we have that $X$ also separates $v$ from $(A - X) \cup C = A' - X$ in $G_{-F}$.
	It follows that there are at most $f$ node-disjoint $A' v$-paths that exclude $F$ in $G$, and so $A' \notpropagate{F}_G B - F$.
	
	Similarly, since $B \notpropagate{F}_H A - F$, we have that for some node $v \in A - F$ there are at most $f$ node-disjoint $Bv$-paths in $H_{-F}$.
	By Menger's Theorem, there exists a cut set $X \not \ni v$ of cardinality at most $f$ that separates $v$ from $B - X$ in $H_{-F}$.
	Since there is no edge from $C$ to $S$ in $G$, we have that $X$ also separates $v$ from $B - X$ in $G_{-F}$.
	Note that $v \in A - F \subseteq A' - F$.
	It follows that there are at most $f$ node-disjoint $B v$-paths that exclude $F$ in $G$, and so $B \notpropagate{F}_G A' - F$.
	This violates condition SC, a contradiction.
\end{proof}

We now show that the paths identified in steps (c), (d), and (e) of the algorithm do indeed exist.
The existence of paths in step (c) follows by construction of $S$.

\begin{lemma} \label{lemma paths}
	For any choice of set $F$ in the algorithm, let $S$ be the unique source component of $G - F$.
	Then, for any two nodes $u \in S \cup \inneighborhood{F}{S}{}$ and $v \in S$, there exists a $uv$-path that excludes $F$.
\end{lemma}

\begin{proof}
	If $u \in S$, then there exists a $uv$-path in $G - F$ since $S$ is strongly connected in $G - F$ by construction.
	If $u \in \inneighborhood{F}{S}{}$, then there exists a node $w \in S$ such that $(u, w)$ is an edge in $G$.
	Also, there exists a $wv$-path $P_{wv}$ in $G - F$ since $S$ is strongly connected in $G - F$.
	So $u \operatorname{-} P_{wv}$ is a $uv$-path in $G$ that excludes $F$.
\end{proof}

The existence of paths in step (d) follows from Lemma \ref{lemma source component satisfies condition SC} and definition of condition SC (Definition \ref{definition condition SC}).

\begin{lemma} \label{lemma propagate within source}
	For any non-faulty node $v$, and any given phase with the corresponding set $F$, in step (d), if $v \in B_v$, then there exist $f+1$ node-disjoint $A_v v$-paths that exclude $F$.
\end{lemma}

\begin{proof}
Fix a phase of the algorithm and the corresponding set $F$.
Consider an arbitrary non-faulty node $v$ such that $v \in B_v$ in step (d).
By construction, $B_v$ is either $N_v - F$ or $Z_v - F$, both of which are non-empty.
In the first case, we have that $A_v = Z_v \propagate{F} N_v - F = B_v$.
In the second case, we have that $Z_v \notpropagate{F} N_v - F$.
By Lemma \ref{lemma source component satisfies condition SC}, we have that $G[Z_v \cup N_v] = G[S \cup \inneighborhood{F}{S}{}]$ satisfies condition SC with parameter $F$.
Therefore $A_v = N_v \propagate{F} Z_v - F = B_v$.
\end{proof}

For paths in step (e), note that, by construction of $S$, the in-neighbors of $S$ are contained entirely in $F$.
So there can only be at most $f$ paths into $S$ from $V - S$.
Since $G$ satisfies condition SC with parameter $F$, we get that $S \propagate{F} V - S - F$.

\begin{lemma} \label{lemma source component propagates}
	For any choice of set $F$ in the algorithm, let $S$ be the unique source component of $G - F$.
	Then $S \propagate{F} V - S - F$.
\end{lemma}

\begin{proof}
Let $A = S$ and $B = F \cup (V - S) = V - A$.
Now, since $A = S$ is the unique source component of $G - F$, we have that $\inneighborhood{B}{A}{} \subseteq F$, i.e., nodes in $F$ are the only ones in $B$ to have an edge into $A$.
So $B$ can have at most $f$ node-disjoint paths to any node in $A$.
Thus $B \notpropagate{F} A - F$.
Since $G$ satisfies condition SC, we have that $S = A \propagate{F} B - F = V - S - F$, as required.
\end{proof} 		\subsection{Proof of Correctness} \label{section proof of correctness}
			We provide a proof of correctness of Algorithm \ref{algorithm consensus} by proving Theorem \ref{theorem sufficiency} below.

\begin{theorem} \label{theorem sufficiency}
	Under the local broadcast model, Byzantine consensus tolerating at most $f$ Byzantine faulty nodes is achievable on a directed graph $G$ if $G$ satisfies condition SC.
\end{theorem}

Let us assume that $G$ satisfies condition SC.
For convenience, we will refer to the state variable $\gamma_v$ as the ``state of node $v$''.
We remind the reader that the algorithm proceeds in phases and each phase has an associated unique set $F$, with $\abs{F} \le f$.
We use $F^*$ to denote the actual set of faulty nodes in a given execution.

As mentioned earlier, the algorithm attempts to balance two objectives.
We formalize them in the two lemmas below.

\begin{lemma} \label{lemma agreement}
	Consider the unique phase of Algorithm \ref{algorithm consensus} where the corresponding set $F = F^*$.
	At the end of this phase, ever pair of non-faulty nodes $u, v \in V$ have identical state, i.e., $\gamma_u = \gamma_v$.
\end{lemma}

\begin{lemma} \label{lemma validity}
	For a non-faulty node $v$, its state $\gamma_v$ at the end of any given phase equals the state of some non-faulty node at the start of that phase.
\end{lemma}

The correctness of Algorithm \ref{algorithm consensus} follows from these two lemmas, as shown next.

\begin{proof}[Proof of Theorem \ref{theorem sufficiency}]
	To prove the correctness of Algorithm \ref{algorithm consensus}, we have to prove the three properties of Agreement, Validity, and Termination, as specified in Section \ref{section notation}.
	\begin{enumerate}[\null]
		\item $\underline{\emph{Termination:}}$
		The algorithm satisfies the termination property because there are a finite number of phases in the algorithm, and each of them completes in finite time.
		
		\item $\underline{\emph{Agreement:}}$
		The total number of faulty nodes is bounded by $f$.
		Therefore, in any execution, there exists at least one phase in which the set $F = F^*$.
		From Lemma \ref{lemma agreement}, all non-faulty nodes have the same state at the end of this phase.
		Lemma \ref{lemma validity} implies that the state of the non-faulty nodes will remain unchanged in the subsequent phases.
		So all non-faulty nodes will have identical outputs, satisfying the agreement property.
		
		\item $\underline{\emph{Validity:}}$
		The state of each non-faulty node is initialized to its own input at the start of the algorithm.
		So the state of each non-faulty node is its own input at the start of the first phase.
		By applying Lemma \ref{lemma validity} inductively, we have that the state of a non-faulty node always equals the \emph{input} of some non-faulty node.
		This satisfies the validity property.
	\end{enumerate}
	This completes the proof of correctness of Algorithm \ref{algorithm consensus}.
\end{proof}

The rest of the section focuses on proving Lemmas \ref{lemma agreement} and \ref{lemma validity}.
We assume that the graph $G$ satisfies condition SC, even if it is not explicitly stated.
We use $F^*$ to denote the actual faulty set.
The following observation follows from the rules used for flooding (Appendix \ref{section flooding}).

\begin{observation} \label{observation fault-free reliable}
	For any phase of Algorithm \ref{algorithm consensus}, for any two nodes $u, v \in V$ (possibly faulty), if $v$ receives value $b$ along a fault-free $uv$-path then $u$ broadcast the value $b$ to its neighbors during flooding.
\end{observation}

Fix a phase in the algorithm along with the corresponding set $F$.
For Lemma \ref{lemma agreement}, the correctness relies on the local broadcast property.
Suppose, as stated in the statement of the lemma, that $F = F^*$ in a given phase.
There are two cases to consider for any non-faulty node $v$.
In the first case $v \in S$.
In this case, the paths used by $v$ in step (c) of the phase exclude $F$.
So these paths are fault-free (i.e., none of their internal nodes are faulty).
Then, the properties of flooding imply that any two non-faulty nodes $u,v$ will obtain $Z_u=Z_v$ and $N_u=N_v$ in step (c).
By a similar argument, all the paths used in step (d) of this phase are also fault-free, and any two non-faulty nodes will end step (d) with an identical state.
In the second case $v \in V - S - F$, a repeat of the above argument implies that the paths used in step (f) are both fault-free and have non-faulty sources.
Since the sources are all from $S$, they have an identical state at the end of step (d).
So $v$ correctly updates its state in step (f) to match the identical state in $S$.

\begin{proof}[Proof of Lemma \ref{lemma agreement}]
	Fix a phase of the algorithm and the corresponding set $F$ such that $F = F^*$.
	We first show that all nodes in $S$ have identical state, at the end of the phase, and then consider nodes in $V - S - F^*$.
	
	Let $Z$ be the set of nodes that flooded $0$ in step (b) and let $N$ be the set of nodes that flooded $1$ in step (b).
	Note that $Z$ and $N$ may contain faulty nodes, but due to the broadcast property, a faulty node is in at most one of these sets.
	Consider any non-faulty node $v \in S$.
	Then $Z_v = Z$ and $N_v = N$, as follows.
	Let $w \in S \cup \inneighborhood{F^*}{S}{}$ be an arbitrary node that flooded $0$ (resp. $1$) in step (b), i.e., $w \in Z$ (resp. $w \in N$).
	Now the $wv$-path $P_{wv}$ identified by $v$ in step (b) excludes $F^*$ and is fault-free.
	So, by Observation \ref{observation fault-free reliable}, $v$ receives $0$ (resp. $1$) along $P_{wv}$ and correctly sets $w \in Z_v$ (resp. $w \in N_v$).
	
	Therefore, we have that for any two non-faulty nodes $u, v \in S$, $Z_u = Z_v = Z$ and $N_u = N_v = N$.
	If $Z - F^*$ (resp. $N - F^*$) is empty, then $S = N - F^*$ (resp. $S = Z - F^*$).
	So all non-faulty nodes in $S$ have identical state, which is not updated, and the claim is trivially true.
	So suppose both $Z - F^*$ and $N - F^*$ are non-empty.
	Since $Z_u = Z_v$ and $N_u = N_v$, we have that $A_u = A_v$ and $B_u = B_v$.
	Let $A := A_u$ and $B := B_u$.
	By construction $A \propagate{F^*} B$.
	Now all nodes in $A$ flooded identical value in step (b), say $\alpha$.
	If $u \in A$, then $u$'s state is $\alpha$ at the beginning of the phase and stays unchanged in step (d), i.e., $\gamma_u = \alpha$.
	If $u \in B$, then the $f+1$ node-disjoint $A u$-paths identified by $u$ in step (d) exclude $F^*$ and so are all fault-free.
	By Observation \ref{observation fault-free reliable}, it follows that $u$ receives $\alpha$ identically along these $f+1$ paths and updates $\gamma_u = \alpha$.
	Similarly, $\gamma_v = \alpha$.
	
	So we have shown that all non-faulty nodes in $S$ have identical state $\alpha$ at the end of step (d).
	Since these nodes do not update their state in the rest of the phase, so we have that all nodes in $S$ have state $\alpha$ at the end of the phase.
	Since all nodes in $S$ are non-faulty, we have that each node in $S$ floods $\alpha$ in step (e).
	Consider any non-faulty node $v \in V - S - F^*$.
	The $f+1$ node-disjoint $Sv$-paths identified by $v$ in step (f) exclude $F^*$ and so are all fault-free.
	Recall that $S$ contains only nodes in $V - F^*$, i.e., only non-faulty nodes, so the source nodes in these $f+1$ paths are also non-faulty.
	By Observation \ref{observation fault-free reliable}, it follows that $v$ receives $\alpha$ identically along these $f+1$ paths and updates $\gamma_v = \alpha$.
	So by the end of the phase, all non-faulty nodes have identical state $\alpha$, as required.
\end{proof}

Lemma \ref{lemma validity} follows from the following observation.
In steps (d) and (f), if a node $v$ updates its state, then it must have received identical value along $f+1$ node-disjoint $A_v v$-paths.
Therefore, at least one of these paths must both be fault-free and have a non-faulty source.

\begin{proof}[Proof of Lemma \ref{lemma validity}]
	Fix a phase of the algorithm and the corresponding set $F$.
	We use $\gamma_u^{\operatorname{start}}$ and $\gamma_u^{\operatorname{end}}$ to denote the state $\gamma_u$ of node $u$ at the beginning and end of the phase, respectively.
	Let $v$ be an arbitrary non-faulty node.
	If $v$ does not update its state in this phase, then $\gamma_u^{\operatorname{end}} = \gamma_u^{\operatorname{start}}$ and the claim is trivially true.
	So suppose that $v$ did update its state in this phase.
	This implies that $v \in V - F$.
	
	There are now two cases to consider
	\begin{enumerate}[\text{Case} 1:]
		\item $v \in S$.
		It follows that $v$ updated its state in step (d).
		Therefore, $v \in B_v$ and $v$ received identical values along $f+1$ node-disjoint $A_v v$-paths in step (b).
		Now, at least one of these paths 1) is fault-free, and 2) has a non-faulty source $u$.
		By Observation \ref{observation fault-free reliable}, we have that the value received by $v$ along this $uv$-path in step (b) is the value flooded by $u$ in step (b).
		So $\gamma_v^{\operatorname{end}} = \gamma_u^{\operatorname{start}}$, where $u$ is a non-faulty node.
		
		\item $v \in V - S - F$.
		It follows that $v$ updated its state in step (f).
		Therefore, $v$ received identical values along $f+1$ node-disjoint $S v$-paths in step (e).
		Now, at least one of these paths 1) is fault-free, and 2) has a non-faulty source $w \in S$.
		By Observation \ref{observation fault-free reliable}, we have that the value received by $v$ along this $wv$-path in step (e) is the value flooded by $w$ in step (e).
		Note that $w \in S$ and so the value flooded by $w$ in step (e) is the state of $w$ at the end of this phase, i.e., $w$ flooded $\gamma_w^{\operatorname{end}}$ in step (e).
		There are further two cases to consider
		\begin{enumerate}[\text{Case} i:]
			\item
			$w$ did not update its state in this iteration.
			Then $\gamma_v^{\operatorname{end}} = \gamma_w^{\operatorname{end}} = \gamma_w^{\operatorname{start}}$.
			Recall that $w$ is a non-faulty node.
			
			\item
			$w$ did update its state in this iteration.
			Then, from Case 1 above, we have that there exists a non-faulty node $u$ such that $\gamma_w^{\operatorname{end}} = \gamma_u^{\operatorname{start}}$.
			So $\gamma_v^{\operatorname{end}} = \gamma_w^{\operatorname{end}} = \gamma_u^{\operatorname{start}}$, where $u$ is a non-faulty node.
		\end{enumerate}
	\end{enumerate}
	In all cases, we have that $\gamma_v^{\operatorname{end}}$ equals $\gamma_u^{\operatorname{start}}$ for some non-faulty node $u$.
\end{proof} 
	\section{Summary} \label{section conclusion}
		In this work, we have presented a tight sufficient and necessary condition for binary valued exact Byzantine consensus in directed graphs under the local broadcast model.
We presented two versions of this condition, condition SC in Definition \ref{definition condition SC} and condition NC in Definition \ref{definition condition NC}.
The sufficiency proof in Section \ref{section sufficiency} is constructive.
However, the algorithm has exponential round complexity.
We leave finding a more efficient algorithm for future work.
The following question is also open: does there exist an efficient algorithm to check if a given directed graph satisfies condition SC? 
	\bibliographystyle{plainurl}
	\bibliography{bib}

	\appendix
	\section{Flooding Procedure} \label{section flooding}
		Any message transmitted during flooding has the form $(b, \Pi)$, where $b \in \set{ 0, 1 }$ and $\Pi$ is a path.
Flooding proceeds in synchronous rounds, with each node possibly forwarding messages received in the previous round, following the rules presented below.
Flooding will end after $n$ rounds, as should be clear from the following description.

To initiate flooding of a value $b \in \set{0, 1}$, a node $v$ transmits message $(b,\bot)$ to its out-neighbors, where $\bot$ represents an empty path.
If $v$ is faulty and does not initiate flooding, then non-faulty out-neighbors of $v$ replace the missing message with the default message of $(1,\bot)$.
Therefore, we can assume that a value is indeed flooded when required by the protocol, by each node even if it is faulty.
When node $v$ receives from an in-neighbor $u$ a message $(\beta,\Pi)$, where $\beta \in \set{0,1}$ and $\Pi$ is a path, it takes the following steps sequentially.
In the following note that $\Pi \operatorname{-} u$ denotes a path obtained by concatenating identifier $u$ to path $\Pi$.
\begin{enumerate}[(i)]
	\item If path $\Pi \operatorname{-} u$ does not exist in graph $G$, then node $v$ discards the message $(\beta,\Pi)$.
	 Recall that each node knows graph $G$, and the message $(\beta,\Pi)$ was received by node $v$ from node $u$.
	\item Else if, in the current phase, node $v$ has previously received from $u$ another message containing path $\Pi$ (i.e., a message of the form $(\beta',\Pi)$),
	then node $v$ discards the message $(\beta,\Pi)$.
	\item Else if path $\Pi$ already includes node $v$'s identifier, node $v$ discards the message $(\beta,\Pi)$.
	\item Else node $v$ is \underline{said to have received value $\beta$ along path $\Pi \operatorname{-} u$} and node $v$ forwards message $(\beta,\Pi \operatorname{-} u)$ to its out-neighbors. Recall that $v$ received message $(\beta,\Pi)$ from its in-neighbor $u$.
\end{enumerate}

Rules (i) and (ii) above are designed to prevent a faulty node from sending spurious messages.
Recall that under the local broadcast model, all out-neighbors of any node receive all its transmissions.
Thus, rule (ii) above essentially prevents a faulty node from successfully delivering mismatching messages to its out-neighbors (i.e., this prevents equivocation).
Rule (iii) ensures that flooding will terminate after $n$ rounds.

Rule (ii) above crucially also ensures the following useful property due to the local broadcast model: even if node $u$ is Byzantine faulty, but paths $P_{uv}$ and $P_{uw}$ are fault-free, then nodes $v$ and $w$ will receive identical value in the message from $u$ forwarded along paths $P_{uv}$ and $P_{uw}$, respectively.
Recall that a path is fault-free if none of the internal nodes on the path are faulty. 
	\section{Proof of Theorem \ref{theorem equivalence}} \label{section proof equivalence}
		Theorem \ref{theorem equivalence} follows from the following lemma.

\begin{lemma} \label{lemma simple equivalence}
	For any set $F$, $G$ satisfies condition NC with parameter $F$ if and only if $G$ satisfies condition SC with parameter $F$.
\end{lemma}
\begin{proof}
	Directly from Lemmas \ref{lemma condition NC implies condition SC} and \ref{lemma condition SC implies condition NC} below.
\end{proof}

\begin{lemma} \label{lemma A propagates B}
	Suppose $G$ satisfies condition NC with parameter $F$.
	For any partition $(A, B)$ of $V(G)$, with both $A - F$ and $B - F$ non-empty, if $B \notadjacent{F} A - F$, then $A \propagate{F} B - F$.
\end{lemma}
\begin{proof}
	Consider an arbitrary set $F$.
	Suppose $G$ satisfies condition NC with parameter $F$ and consider a partition $(A, B)$ with $A - F$ and $B - F$ non-empty.
	Suppose, for the sake of contradiction that $B \notadjacent{F}_G A - F$ and $A \notpropagate{F}_G B - F$.
	Obtain $G'$ from $G_{-F}$ by adding a node $s$ with no incoming edges but an outgoing edge to each node in $A$.
	Now $A \notpropagate{F}_G B - F$ so that there exists a node $t \in B - F$ such that there are at most $f$ node-disjoint $At$-paths in $G$ that exclude $F$.
	Recall that this implies at most $f$ node-disjoint $At$-paths in $G_{-F}$.
	Therefore, there are at most $f$ node-disjoint $st$-paths in $G'$.
	By Menger's Theorem, there exists a cut set $X$, of size at most $f$, that partitions $G' - X$ into $(S, T)$ with $s \in S$, $t \in T$ and no edge from $S$ to $T$.
	Observe that $T \subseteq B$, since otherwise we have an edge from $s$ to a node in $T$.
	
	We now create a partition $(L, C, R)$ that violates the requirements of condition NC, achieving the desired contradiction.
	Let $L = A$, $C = B - T$, and $R = B \cap T = T$.
	Note that $L - F = A - F$ is non-empty and $t \in T - F = R - F$ and so $R - F$ is non-empty as well.
	Now $R \cup C = B \notadjacent{F}_G A - F = L - F$.
	We show that $L \cup C \notadjacent{F}_G R - F$, completing the contradiction.
	By the property of the $st$-cut $X$, $T$ can only have in-neighbors in $X$, since there is no edge from $S$ to $T$.
	Now, the in-neighbors of $R - F$ in $G$ are the same as the in-neighbors of $R - F$ in $G'$ since no edge was added or removed from any node in $R - F = T - F$.
	Therefore, $T - F$ can only have in-neighbors in $X$ and $T$ in the graph $G$.
	So the in-neighbors of $T - F$ in $L \cup C = V - T$ are contained entirely in $X$ and we have that $\abs[1]{ \inneighborhood{L \cup C}{T - F}{G} } \le \abs{X} \le f$.
	That is $L \cup C \notadjacent{F}_G T - F = R - F$, as required.
\end{proof}

\begin{lemma} \label{lemma B does not propagate A}
	Suppose $G$ satisfies condition NC with parameter $F$.
	For any partition $(A, B)$ of $V(G)$, with both $A - F$ and $B - F$ non-empty, if $B \notpropagate{F} A - F$, then there exists a partition $(A', B')$ of $V(G)$ such that
	\begin{itemize}
		\item $A' - F$ and $B' - F$ are both non-empty,
		\item $A' \subseteq A$ and $B \subseteq B'$, and
		\item $B' \notadjacent{F} A' - F$.
	\end{itemize}
\end{lemma}
\begin{proof}
	Consider an arbitrary set $F$.
	Suppose $G$ satisfies condition NC with parameter $F$ and a partition $(A, B)$ with $A - F$ and $B - F$ non-empty.
	Suppose that $B \notpropagate{F}_G A - F$.
	Obtain $G'$ from $G_{-F}$ by adding a node $s$ with no incoming edges but an outgoing edge to each node in $B$.
	Since $B \notpropagate{F}_G A - F$, there exists a node $t \in A - F$ such that there are at most $f$ node-disjoint $st$-paths in $G'$.
	By Menger's Theorem, there exists a cut set $X$, of size at most $f$, that partitions $G' - X$ into $(S, T)$ with $s \in S$, $t \in T$ and no edge from $S$ to $T$.
	Observe that $T \subseteq A$, since otherwise we have an edge from $s$ to a node in $T$.
	
	We now create a partition $(A', B')$ of $V(G)$.
	Let $B' = (S - s) \cup X \supseteq B$ (since $B = V - A \subseteq V - T = (S - s) \cup X$) and $A' = T \subseteq A$.
	Observe that $V = T \cup X \cup S - s = A' \cup B'$ and $A' \cap B' = \emptyset$, so that $(A', B')$ is a partition of $V$.
	Since $B - F$ is non-empty, so $B' - F$ is also non-empty.
	$A'$ is non-empty since $t \in T - F = A' - F$.
	By the property of the $st$-cut $X$, $T$ can only have in-neighbors in $X$, since there is no edge from $S$ to $T$.
	Now $A' - F = T - F$ had no edges removed or added in $G'$ so its in-neighbors are the same in $G$ and $G'$.
	Therefore, $T - F$ can only have in-neighbors in $X$ and $T$ in the graph $G$.
	So the in-neighbors of $T - F$ in $B' = V - T$ are contained entirely in $X$ and we have that $\abs[1]{ \inneighborhood{B'}{T - F}{G} } \le \abs{X} \le f$.
	That is $B' \notadjacent{F}_G T - F = A' - F$, as required.
\end{proof}

\begin{lemma} \label{lemma condition NC implies condition SC}
	For any set $F$, if $G$ satisfies condition NC with parameter $F$, then $G$ satisfies condition SC with parameter $F$.
\end{lemma}
\begin{proof}
	Consider an arbitrary set $F$, a graph $G$ that satisfies condition NC with parameter $F$.
	Let $(A, B)$ be any partition such that $A - F$ and $B - F$ are both non-empty.
	If $B \propagate{F}_{G} A - F$, then we are done.
	So suppose $B \notpropagate{F}_{G} A - F$.
	Then, by Lemma \ref{lemma B does not propagate A}, there exists a partition $(A', B')$ as specified in the statement of Lemma \ref{lemma B does not propagate A}.
	Since $B' \notadjacent{F}_G A' - F$, by Lemma \ref{lemma A propagates B} we have that $A' \propagate{F}_G B' - F$.
	Since $A' \subseteq A$ and $B \subseteq B'$, it follows that $A \propagate{F}_G B - F$, as required.
\end{proof}

\begin{lemma} \label{lemma condition SC implies condition NC}
	For any set $F$, if $G$ satisfies condition SC with parameter $F$, then $G$ satisfies condition NC with parameter $F$.
\end{lemma}
\begin{proof}
	Fix an arbitrary set $F$.
	We show the contrapositive that if $G$ does not satisfy condition NC with parameter $F$, then $G$ does not satisfy condition SC with parameter $F$.
	Let $(L, C, R)$ be a partition that violates the requirements of condition NC.
	Let $A = L$ and $B = R \cup C$.
	Then $A - F = L - F$ is non-empty as well as $B - F \supseteq R - F$.
	
	We first show that $B \notpropagate{F}_G A - F$.
	Note that any path terminating in $V - F$ and excluding $F$ is also a path in $G_{-F}$.
	So we show that there are at most $f$ node-disjoint $Bt$-paths for every node $t \in A - F$ in $G_{-F}$, i.e., $B \notpropagate{F}_{G_{-F}} A - F$.
	Obtain $G'$ from $G_{-F}$ by adding a node $s$ with no incoming edges but an outgoing edge to each node in $B = R \cup C$.
	Consider any node $t \in A - F = L - F$.
	Then $\inneighborhood{R \cup C}{L - F}{G'} = \inneighborhood{R \cup C}{L - F}{G}$ is an $st$-cut in $G'$.
	Since $R \cup C \notadjacent{F}_G L - F$, we have that $\abs{\inneighborhood{R \cup C}{L - F}{G}} \le f$.
	So by Menger's Theorem, at most $f$ node-disjoint $st$-paths exist in $G'$.
	Therefore, there are at most $f$ node-disjoint paths from $B = R \cup C$ to $t \in A - F = L - F$ in $G_{-F}$, i.e., $B \notpropagate{F}_{G_{-F}} A - F$.
	
    We now show that $A = L \notpropagate{F}_G R - F \subseteq B - F$ using a similar argument.
    As before, we show that there are at most $f$ node-disjoint $At$-paths for every node $t \in R - F \subseteq B - F$ in $G_{-F}$, i.e., $A \notpropagate{F}_{G_{-F}} R - F \subseteq B - F$.
	Obtain $G'$ from $G_{-F}$ by adding a node $s$ with no incoming edges but an outgoing edge to each node in $A = L$.
	Consider any node $t \in R - F \subseteq B - F$.
	Then $\inneighborhood{L \cup C}{R - F}{G'} = \inneighborhood{L \cup C}{R - F}{G}$ is an $st$-cut in $G'$.
	Since $L \cup C \notadjacent{F}_G R - F$, we have that $\abs{\inneighborhood{L \cup C}{R - F}{G}} \le f$.
	So by Menger's Theorem, at most $f$ node-disjoint $st$-paths exist in $G'$.
	Therefore, there are at most $f$ node-disjoint paths from $A = L$ to $t \in R - F \subseteq B - F$ in $G_{-F}$, i.e., $A \notpropagate{F}_{G_{-F}} B - F$.
\end{proof} 
	\section{Proof of Theorem \ref{theorem necessity}} \label{section proof necessity}
		\begin{proof}[Proof of Theorem \ref{theorem necessity}]
	\begin{figure}[htb]
		\centering
		\begin{tikzpicture}[scale=0.75, every node/.style={scale=0.75}]
		\node[draw, circle, minimum size=1cm] at (0, 0) (L_1 cap F) {$L_1 \cap F_1$};
		\node[draw, circle, minimum size=1cm] at (8, 0) (R_1 cap F) {$R_1 \cap F_1$};
		\node[draw, circle, minimum size=1cm] at (0, -3) (L_1) {$L_1 - F_1$};
		\node[draw, circle, minimum size=1cm] at (8, -3) (R_1) {$R_1 - F_1$};
		
		\node[draw, ellipse, minimum height=2cm] at (0, -6) (NL cap FR) {$\inneighborhood{L \cap F}{R-F}{}$};
		\node[draw, ellipse, minimum height=2cm] at (8, -6) (NR cap FL) {$\inneighborhood{R \cap F}{L-F}{}$};
		\node[draw, ellipse, minimum height=2cm] at (0, -9) (NL minus FR) {$\inneighborhood{L - F}{R-F}{}$};
		\node[draw, ellipse, minimum height=2cm] at (8, -9) (NR minus FL) {$\inneighborhood{R - F}{L-F}{}$};
		
		\node[draw, circle, minimum size=1cm] at (0, -12) (L_0) {$L_0 - F_0$};
		\node[draw, circle, minimum size=1cm] at (8, -12) (R_0) {$R_0 - F_0$};
		\node[draw, circle, minimum size=1cm] at (0, -15) (L_0 cap F) {$L_0 \cap F_0$};
		\node[draw, circle, minimum size=1cm] at (8, -15) (R_0 cap F) {$R_0 \cap F_0$};
		
		\draw[-{Latex[width=3mm,length=3mm]},out=180, in=180,red] (NL cap FR) to (L_1 cap F);
		\draw[-{Latex[width=3mm,length=3mm]},out=180, in=180,red] (NL minus FR) to (L_1 cap F);
		\draw[-{Latex[width=3mm,length=3mm]},out=180, in=180,red] (NL cap FR) to (L_1);
		\draw[-{Latex[width=3mm,length=3mm]},out=180, in=180,red] (NL minus FR) to (L_1);
		
		\draw[-,out=0, in=0] (NR cap FL) to (R_1 cap F);
		\draw[-,out=0, in=0] (NR minus FL) to (R_1 cap F);
		\draw[-,out=0, in=0] (NR cap FL) to (R_1);
		\draw[-,out=0, in=0] (NR minus FL) to (R_1);
		
		\draw[-] (L_1 cap F) to (L_1);
		\draw[-] (R_1 cap F) to (R_1);
		
		\draw[-] (L_1 cap F) to (R_1 cap F);
		\draw[-] (NR cap FL) to (L_1 cap F);
		\draw[-{Latex[width=3mm,length=3mm]}] (NR minus FL) to (L_1 cap F);
		\draw[-{Latex[width=3mm,length=3mm]}] (R_1) to (L_1 cap F);
		
		\draw[-{Latex[width=3mm,length=3mm]},red] (NL cap FR) to (R_1 cap F);
		\draw[-{Latex[width=3mm,length=3mm]}] (NL minus FR) to (R_1 cap F);
		\draw[-{Latex[width=3mm,length=3mm]}] (L_1) to (R_1 cap F);
		
		\draw[-] (NR cap FL) to (L_1);
		\draw[-{Latex[width=3mm,length=3mm]}] (NR minus FL) to (L_1);
		
		\draw[-{Latex[width=3mm,length=3mm]},red] (NL cap FR) to (R_1);
		\draw[-{Latex[width=3mm,length=3mm]}] (NL minus FR) to (R_1);
		
		\draw[-] (NR cap FL) to (NR minus FL);
		\draw[-] (NL cap FR) to (NL minus FR);
		\draw[-] (NL cap FR) to (NR cap FL);
		\draw[-] (NL minus FR) to (NR minus FL);
		\draw[-] (NL cap FR) to (NR minus FL);
		\draw[-] (NR cap FL) to (NL minus FR);
		
		\draw[-{Latex[width=3mm,length=3mm]},out=0, in=0,red] (NR cap FL) to (R_0 cap F);
		\draw[-{Latex[width=3mm,length=3mm]},out=0, in=0,red] (NR minus FL) to (R_0 cap F);
		\draw[-{Latex[width=3mm,length=3mm]},out=0, in=0,red] (NR cap FL) to (R_0);
		\draw[-{Latex[width=3mm,length=3mm]},out=0, in=0,red] (NR minus FL) to (R_0);
		
		\draw[-,out=180, in=180] (NL cap FR) to (L_0 cap F);
		\draw[-,out=180, in=180] (NL minus FR) to (L_0 cap F);
		\draw[-,out=180, in=180] (NL cap FR) to (L_0);
		\draw[-,out=180, in=180] (NL minus FR) to (L_0);
		
		\draw[-] (L_0 cap F) to (L_0);
		\draw[-] (R_0 cap F) to (R_0);
		
		\draw[-] (L_0 cap F) to (R_0 cap F);
		\draw[-{Latex[width=3mm,length=3mm]},red] (NR cap FL) to (L_0 cap F);
		\draw[-{Latex[width=3mm,length=3mm]}] (NR minus FL) to (L_0 cap F);
		\draw[-{Latex[width=3mm,length=3mm]}] (R_0) to (L_0 cap F);
		
		\draw[-] (NL cap FR) to (R_0 cap F);
		\draw[-{Latex[width=3mm,length=3mm]}] (NL minus FR) to (R_0 cap F);
		\draw[-{Latex[width=3mm,length=3mm]}] (L_0) to (R_0 cap F);
		
		\draw[-{Latex[width=3mm,length=3mm]},red] (NR cap FL) to (L_0);
		\draw[-{Latex[width=3mm,length=3mm]}] (NR minus FL) to (L_0);
		
		\draw[-] (NL cap FR) to (R_0);
		\draw[-{Latex[width=3mm,length=3mm]}] (NL minus FR) to (R_0);
		\end{tikzpicture}
		\caption{Network $\mathcal{G}$: edges between copies of nodes in $L$ and $R$.}
		\label{figure necessity}
	\end{figure}
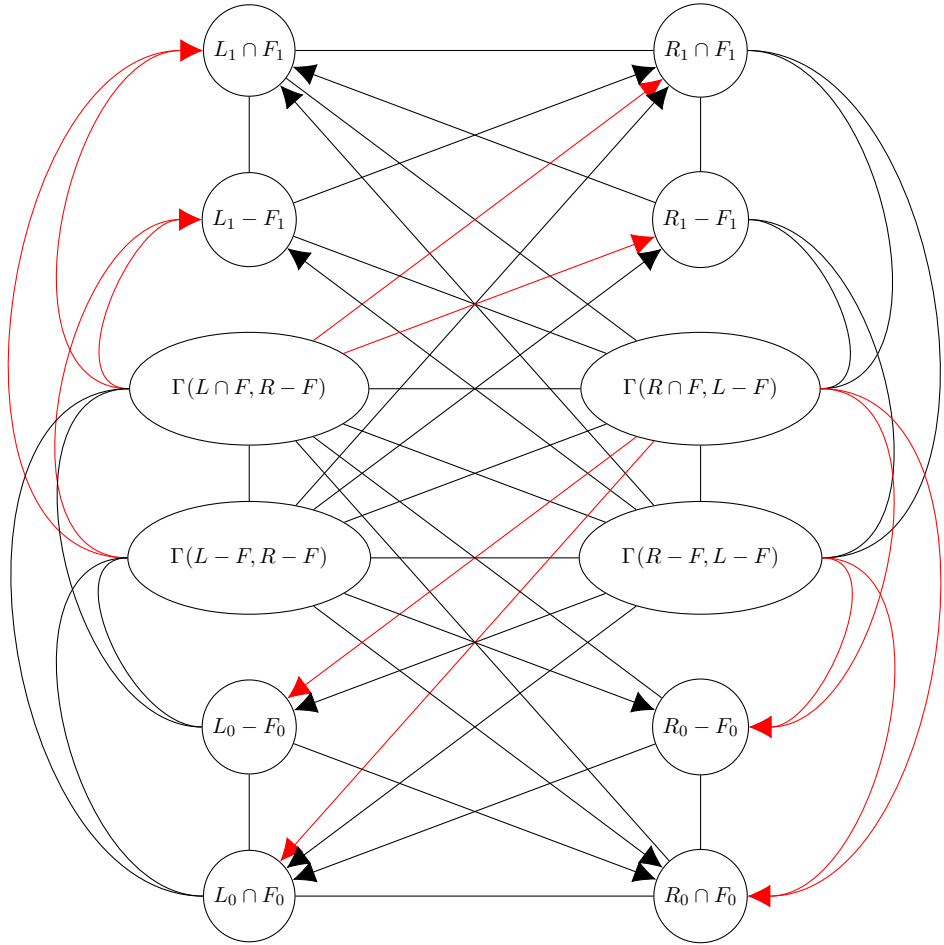
	
	\begin{figure}[htb]
		\centering
		\begin{tikzpicture}[scale=0.75, every node/.style={scale=0.75}]
		\node[draw, circle, minimum size=1cm] at (0, 0) (L_1 cap F) {$L_1 \cap F_1$};
		\node[draw, circle, minimum size=1cm] at (0, -3) (L_1) {$L_1 - F_1$};
		\node[draw, ellipse, minimum height=2cm] at (0, -6) (NL cap FR) {$\inneighborhood{L \cap F}{R-F}{}$};
		\node[draw, ellipse, minimum height=2cm] at (0, -9) (NL minus FR) {$\inneighborhood{L - F}{R-F}{}$};
		\node[draw, circle, minimum size=1cm] at (0, -12) (L_0) {$L_0 - F_0$};
		\node[draw, circle, minimum size=1cm] at (0, -15) (L_0 cap F) {$L_0 \cap F_0$};
		
		\node[draw, circle, minimum size=1cm, label=right:{$\del{\inneighborhood{C}{R-F}{} - \inneighborhood{C}{L - F}{}}_1$}] at (6, -4.5) (NCR_1) {};
		\node[draw, circle, minimum size=1cm, label=right:{$\del{\inneighborhood{C}{L - F}{} - \inneighborhood{C}{R-F}{}}_1$}] at (6, -1.5) (NCL_1) {};
		\node[draw, circle, minimum size=1cm, label=right:{$\inneighborhood{C}{L - F}{} \cap \inneighborhood{C}{R-F}{}$}] at (6, -7.5) (NCLR) {};
		\node[draw, circle, minimum size=1cm, label=right:{$\del{\inneighborhood{C}{L - F}{} - \inneighborhood{C}{R-F}{}}_0$}] at (6, -13.5) (NCL_0) {};
		\node[draw, circle, minimum size=1cm, label=right:{$\del{\inneighborhood{C}{R-F}{} - \inneighborhood{C}{L - F}{}}_0$}] at (6, -10.5) (NCR_0) {};
		
		\node[draw, circle, minimum size=1cm] at (-5, -11.5) (C_0) {$C_0$};
		\node[draw, circle, minimum size=1cm] at (-5, -7.5) (C'_1) {$C'_1$};
		\node[draw, circle, minimum size=1cm] at (-5, -3.5) (C_1) {$C_1$};
		
		\draw[-] (L_1 cap F) to (C_1);
		\draw[-{Latex[width=3mm,length=3mm]}] (L_1) to (C_1);
		\draw[-] (L_0 cap F) to (C_0);
		\draw[-{Latex[width=3mm,length=3mm]}] (L_0) to (C_0);
		\draw[-{Latex[width=3mm,length=3mm]},red] (L_0 cap F) to (C'_1);
		\draw[-{Latex[width=3mm,length=3mm]}] (L_0) to (C'_1);
		\draw[-] (NL cap FR) to (C_0);
		\draw[-{Latex[width=3mm,length=3mm]},red] (NL cap FR) to (C_1);
		\draw[-{Latex[width=3mm,length=3mm]},red] (NL cap FR) to (C'_1);
		\draw[-{Latex[width=3mm,length=3mm]}] (NL minus FR) to (C_1);
		\draw[-{Latex[width=3mm,length=3mm]}] (NL minus FR) to (C'_1);
		\draw[-{Latex[width=3mm,length=3mm]}] (NL minus FR) to (C_0);
		
		\draw[-] (L_0) to (NCL_0);
		\draw[-] (L_0) to (NCLR);
		\draw[-{Latex[width=3mm,length=3mm]}] (L_0) to (NCR_0);
		\draw[-{Latex[width=3mm,length=3mm]}] (L_0) to (NCR_1);
		\draw[-] (L_0 cap F) to (NCL_0);
		\draw[-] (L_0 cap F) to (NCLR);
		\draw[-] (L_0 cap F) to (NCR_0);
		\draw[-{Latex[width=3mm,length=3mm]},red] (L_0 cap F) to (NCR_1);
		
		\draw[-] (L_1) to (NCL_1);
		\draw[-{Latex[width=3mm,length=3mm]},red] (NCLR) to (L_1);
		\draw[-] (L_1 cap F) to (NCL_1);
		\draw[-{Latex[width=3mm,length=3mm]},red] (NCR_1) to (L_1 cap F);
		\draw[-{Latex[width=3mm,length=3mm]},red] (NCLR) to (L_1 cap F);
		
		\draw[-] (NL cap FR) to (NCL_0);
		\draw[-{Latex[width=3mm,length=3mm]},red] (NL cap FR) to (NCL_1);
		\draw[-] (NL cap FR) to (NCR_0);
		\draw[-{Latex[width=3mm,length=3mm]},red] (NL cap FR) to (NCR_1);
		\draw[-] (NL cap FR) to (NCLR);
		
		\draw[-] (NL minus FR) to (NCL_0);
		\draw[-{Latex[width=3mm,length=3mm]},red] (NL minus FR) to (NCL_1);
		\draw[-{Latex[width=3mm,length=3mm]}] (NL minus FR) to (NCR_0);
		\draw[-{Latex[width=3mm,length=3mm]}] (NL minus FR) to (NCR_1);
		\draw[-] (NL minus FR) to (NCLR);
		\end{tikzpicture}
		\caption{Network $\mathcal{G}$: edges between copies of nodes in $L$ and $C$.}
		\label{figure necessity LC}
	\end{figure}
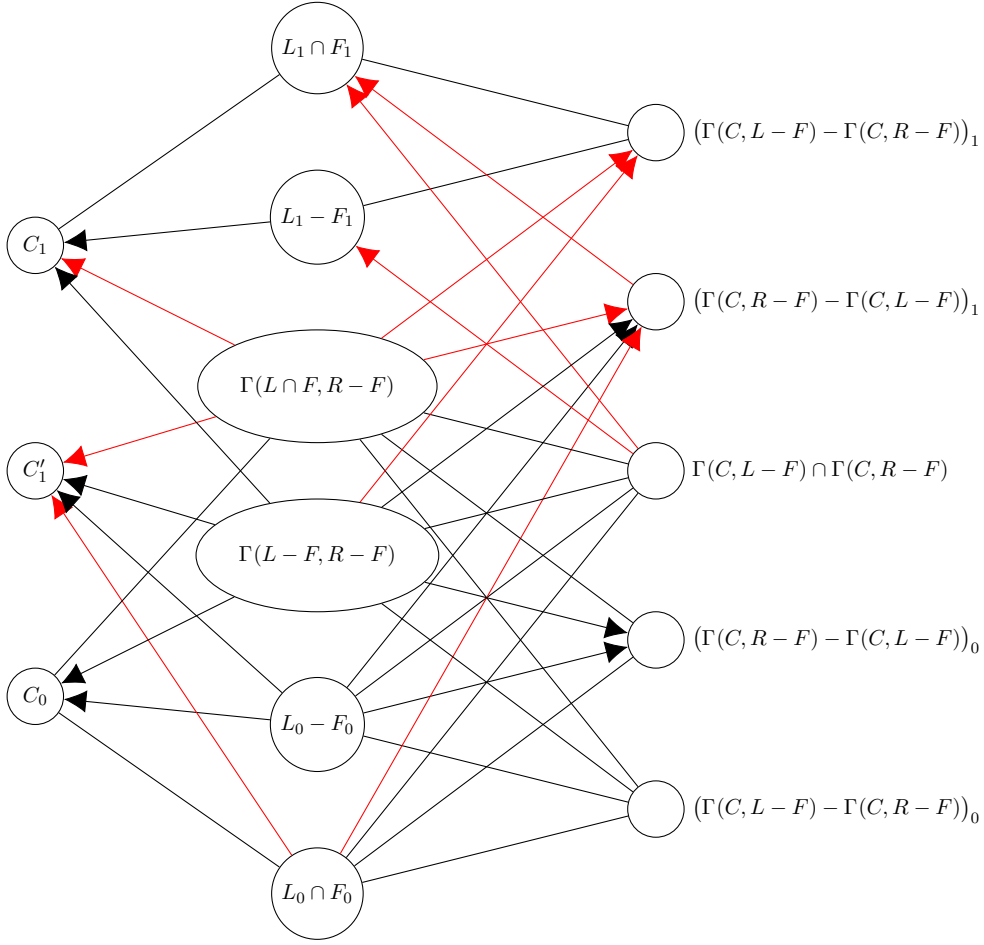
	
	\begin{figure}[htb]
		\centering
		\begin{tikzpicture}[scale=0.75, every node/.style={scale=0.75}]
		\node[draw, circle, minimum size=1cm] at (0, 0) (R_1 cap F) {$R_1 \cap F_1$};
		\node[draw, circle, minimum size=1cm] at (0, -3) (R_1) {$R_1 - F_1$};
		\node[draw, ellipse, minimum height=2cm] at (0, -6) (NR cap FL) {$\inneighborhood{R \cap F}{L - F}{}$};
		\node[draw, ellipse, minimum height=2cm] at (0, -9) (NR minus FL) {$\inneighborhood{R - F}{L - F}{}$};
		\node[draw, circle, minimum size=1cm] at (0, -12) (R_0) {$R_0 - F_0$};
		\node[draw, circle, minimum size=1cm] at (0, -15) (R_0 cap F) {$R_0 \cap F_0$};
		
		\node[draw, circle, minimum size=1cm, label=right:{$\del{\inneighborhood{C}{R - F}{} - \inneighborhood{C}{L - F}{}}_1$}] at (6, -4.5) (NCR_1) {};
		\node[draw, circle, minimum size=1cm, label=right:{$\del{\inneighborhood{C}{L - F}{} - \inneighborhood{C}{R - F}{}}_1$}] at (6, -1.5) (NCL_1) {};
		\node[draw, circle, minimum size=1cm, label=right:{$\inneighborhood{C}{L - F}{} \cap \inneighborhood{C}{R - F}{}$}] at (6, -7.5) (NCLR) {};
		\node[draw, circle, minimum size=1cm, label=right:{$\del{\inneighborhood{C}{L - F}{} - \inneighborhood{C}{R - F}{}}_0$}] at (6, -13.5) (NCL_0) {};
		\node[draw, circle, minimum size=1cm, label=right:{$\del{\inneighborhood{C}{R - F}{} - \inneighborhood{C}{L - F}{}}_0$}] at (6, -10.5) (NCR_0) {};
		
		\node[draw, circle, minimum size=1cm] at (-5, -11.5) (C_0) {$C_0$};
		\node[draw, circle, minimum size=1cm] at (-5, -7.5) (C'_1) {$C'_1$};
		\node[draw, circle, minimum size=1cm] at (-5, -3.5) (C_1) {$C_1$};
		
		\draw[-] (R_1 cap F) to (C_1);
		\draw[-{Latex[width=3mm,length=3mm]}] (R_1) to (C_1);
		\draw[-] (R_0 cap F) to (C_0);
		\draw[-{Latex[width=3mm,length=3mm]}] (R_0) to (C_0);
		\draw[-{Latex[width=3mm,length=3mm]},red] (R_1 cap F) to (C'_1);
		\draw[-{Latex[width=3mm,length=3mm]}] (R_1) to (C'_1);
		\draw[-] (NR cap FL) to (C_1);
		\draw[-{Latex[width=3mm,length=3mm]},red] (NR cap FL) to (C_0);
		\draw[-{Latex[width=3mm,length=3mm]},red] (NR cap FL) to (C'_1);
		\draw[-{Latex[width=3mm,length=3mm]}] (NR minus FL) to (C_1);
		\draw[-{Latex[width=3mm,length=3mm]}] (NR minus FL) to (C'_1);
		\draw[-{Latex[width=3mm,length=3mm]}] (NR minus FL) to (C_0);
		
		\draw[-] (R_0) to (NCR_0);
		\draw[-{Latex[width=3mm,length=3mm]},red] (NCLR) to (R_0);
		\draw[-{Latex[width=3mm,length=3mm]},red] (NCLR) to (R_0 cap F);
		\draw[-{Latex[width=3mm,length=3mm]},red] (R_1 cap F) to (NCL_0);
		\draw[-] (R_0 cap F) to (NCR_0);
		
		\draw[-] (R_1) to (NCR_1);
		\draw[-] (R_1) to (NCLR);
		\draw[-{Latex[width=3mm,length=3mm]}] (R_1) to (NCL_1);
		\draw[-{Latex[width=3mm,length=3mm]}] (R_1) to (NCL_0);
		\draw[-{Latex[width=3mm,length=3mm]},red] (NCL_0) to (R_0 cap F);
		\draw[-] (R_1 cap F) to (NCR_1);
		\draw[-] (R_1 cap F) to (NCLR);
		\draw[-] (NCL_1) to (L_1 cap F);
		
		\draw[-] (NR cap FL) to (NCR_1);
		\draw[-{Latex[width=3mm,length=3mm]},red] (NR cap FL) to (NCR_0);
		\draw[-] (NR cap FL) to (NCL_1);
		\draw[-{Latex[width=3mm,length=3mm]},red] (NR cap FL) to (NCL_0);
		\draw[-] (NR cap FL) to (NCLR);
		
		\draw[-] (NR minus FL) to (NCR_1);
		\draw[-{Latex[width=3mm,length=3mm]},red] (NR minus FL) to (NCR_0);
		\draw[-{Latex[width=3mm,length=3mm]}] (NR minus FL) to (NCL_0);
		\draw[-{Latex[width=3mm,length=3mm]}] (NR minus FL) to (NCL_1);
		\draw[-] (NR minus FL) to (NCLR);
		\end{tikzpicture}
		\caption{Network $\mathcal{G}$: edges between copies of nodes in $R$ and $C$.}
		\label{figure necessity RC}
	\end{figure}
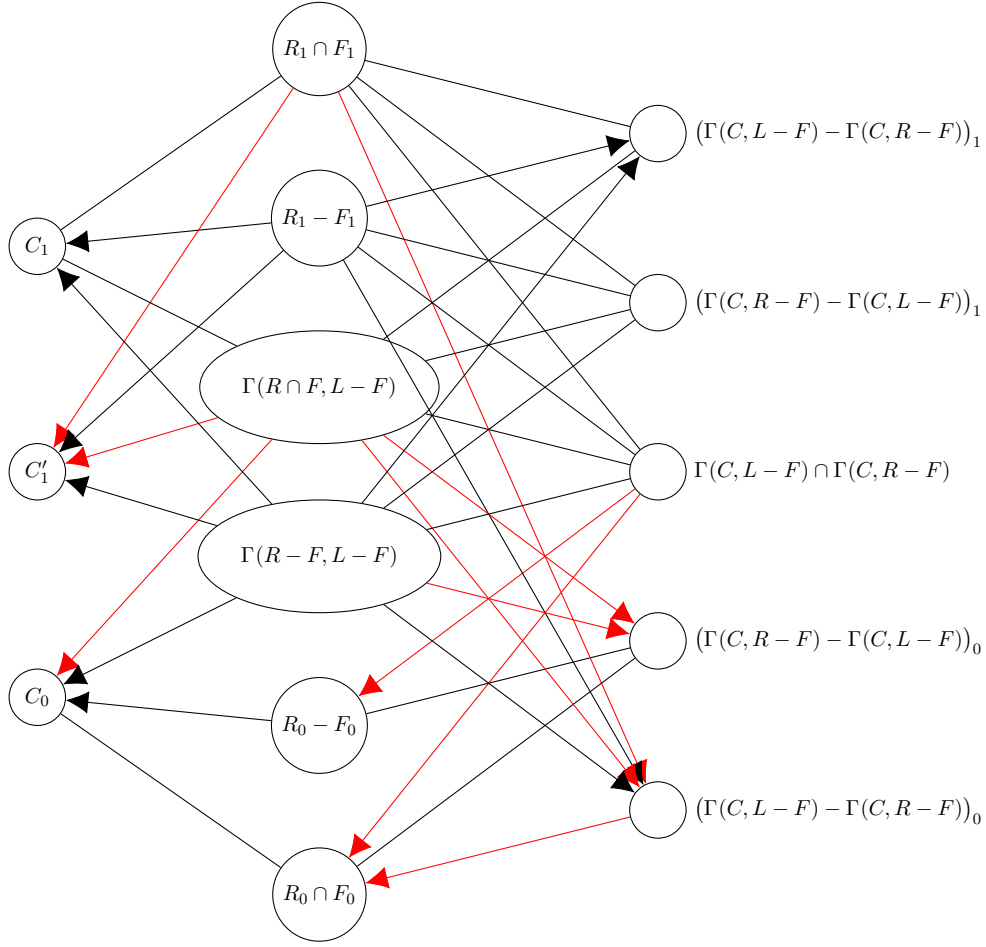
	
	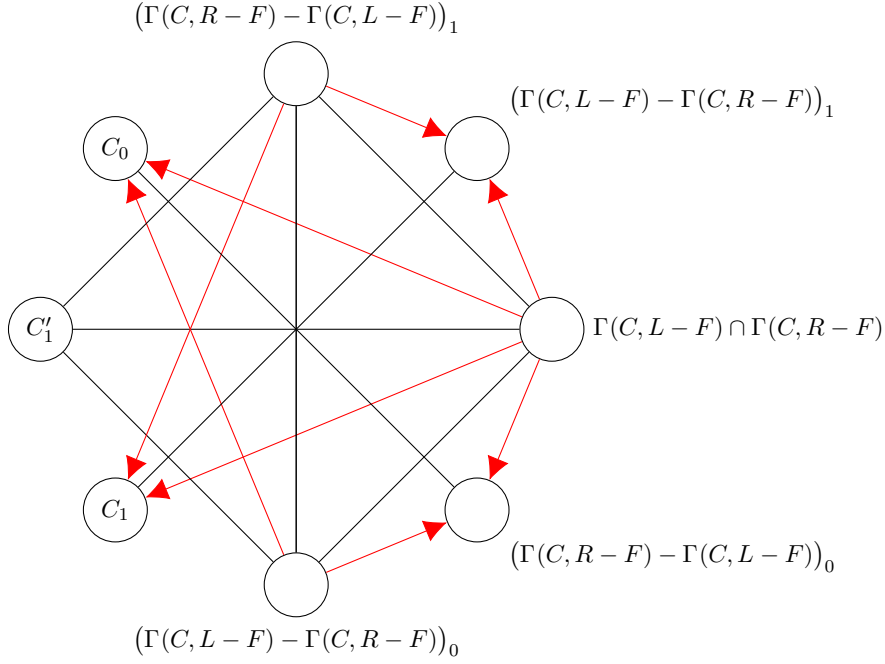
\begin{figure}[htb]
		\centering
		\begin{tikzpicture}[scale=0.85, every node/.style={scale=0.85}]
		\node[draw, circle, minimum size=1cm] at (135:4) (C_0) {$C_0$};
		\node[draw, circle, minimum size=1cm] at (180:4) (C'_1) {$C'_1$};
		\node[draw, circle, minimum size=1cm] at (225:4) (C_1) {$C_1$};
		\node[draw, circle, minimum size=1cm, label=below:{$\del{\inneighborhood{C}{L - F}{} - \inneighborhood{C}{R - F}{}}_0$}] at (270:4) (NCL_0) {};
		\node[draw, circle, minimum size=1cm, label=below right:{$\del{\inneighborhood{C}{R - F}{} - \inneighborhood{C}{L - F}{}}_0$}] at (315:4) (NCR_0) {};
		\node[draw, circle, minimum size=1cm, label=right:{$\inneighborhood{C}{L - F}{} \cap \inneighborhood{C}{R - F}{}$}] at (360:4) (NCLR) {};
		\node[draw, circle, minimum size=1cm, label=above right:{$\del{\inneighborhood{C}{L - F}{} - \inneighborhood{C}{R - F}{}}_1$}] at (45:4) (NCL_1) {};
		\node[draw, circle, minimum size=1cm, label=above:{$\del{\inneighborhood{C}{R - F}{} - \inneighborhood{C}{L - F}{}}_1$}] at (90:4) (NCR_1) {};
		
		\draw[-{Latex[width=3mm,length=3mm]},red] (NCLR) to (C_0);
		\draw[-] (NCLR) to (C'_1);
		\draw[-{Latex[width=3mm,length=3mm]},red] (NCLR) to (C_1);
		\draw[-] (NCLR) to (NCL_0);
		\draw[-{Latex[width=3mm,length=3mm]},red] (NCLR) to (NCL_1);
		\draw[-{Latex[width=3mm,length=3mm]},red] (NCLR) to (NCR_0);
		\draw[-] (NCLR) to (NCR_1);
		
		\draw[-{Latex[width=3mm,length=3mm]},red] (NCL_0) to (C_0);
		\draw[-] (NCL_0) to (C'_1);
		\draw[-{Latex[width=3mm,length=3mm]},red] (NCL_0) to (NCR_0);
		\draw[-] (NCL_0) to (NCR_1);
		
		\draw[-{Latex[width=3mm,length=3mm]},red] (NCR_1) to (C_1);
		\draw[-] (NCR_1) to (C'_1);
		\draw[-{Latex[width=3mm,length=3mm]},red] (NCR_1) to (NCL_1);
		\draw[-] (NCR_1) to (NCL_0);
		
		\draw[-] (NCR_0) to (C_0);
		\draw[-] (NCL_1) to (C_1);
		\end{tikzpicture}
		\caption{Network $\mathcal{G}$: edges between copies of nodes in $C$.}
		\label{figure necessity C}
	\end{figure}
	
	Suppose for the sake of contradiction that $G$ does not satisfy condition NC and there exists an algorithm $\mathcal{A}$ that solves Byzantine consensus under the local broadcast model.
	The Byzantine consensus algorithm $\mathcal{A}$ outlines a procedure $\mathcal{A}_u$ for each node $u \in V$ that describes state transitions of $u$.
	In each synchronous round, each node $u$ optionally sends messages to its out-neighbors, receives messages from its in-neighbors, and then updates its state.
	The new state of $u$ depends entirely on $u$'s previous state and the messages received by $u$ from its neighbors.
	The state of $u$ determines the messages sent by $u$.
	
	We first create a network $\mathcal{G}$ to model behavior of nodes in $G$ in three different executions $E_1$, $E_2$, and $E_3$, which we will describe later.
	Figures \ref{figure necessity}-\ref{figure necessity C} depict the network $\mathcal{G}$.
	A node $u$ in $G$ may have multiple copies in $\mathcal{G}$, as listed below.
	
	\begin{enumerate}[1)]
		\item
		If $u \in L - \inneighborhood{L}{R - F}{}$, then $u$ has two copies $u_0$ and $u_1$.
		$u_0$ starts with input $0$ while $u_1$ starts with input $1$.
		$L_0$ and $L_1$ denote the two sets corresponding to the two copies of nodes in $L - \inneighborhood{L}{R - F}{}$.
		
		\item
		If $u \in R - \inneighborhood{R}{L - F}{}$, then $u$ has two copies $u_0$ and $u_1$.
		$u_0$ starts with input $0$ while $u_1$ starts with input $1$.
		$R_0$ and $R_1$ denote the two sets corresponding to the two copies of nodes in $R - \inneighborhood{R}{L - F}{}$.
		
		\item
		If $u \in \inneighborhood{L}{R - F}{}$, then $u$ has a single copy that starts with input $0$.
		
		\item
		If $u \in \inneighborhood{R}{L - F}{}$, then $u$ has a single copy that starts with input $1$.
		
		\item
		If $u \in \inneighborhood{C}{L - F}{} \cap \inneighborhood{C}{R - F}{}$, then $u$ has a single copy that starts with input $1$.
		
		\item
		If $u \in \inneighborhood{C}{L - F}{} - \inneighborhood{C}{R - F}{}$, then $u$ has two copies $u_0$ and $u_1$.
		$u_0$ starts with input $0$ while $u_1$ starts with input $1$.
		We denote the two sets corresponding to the two copies of nodes in $\inneighborhood{C}{L - F}{} - \inneighborhood{C}{R - F}{}$ by $\del[1]{\inneighborhood{C}{L - F}{} - \inneighborhood{C}{R - F}{}}_0$ and $\del[1]{\inneighborhood{C}{L - F}{} - \inneighborhood{C}{R - F}{}}_1$.
		
		\item
		If $u \in \inneighborhood{C}{R - F}{} - \inneighborhood{C}{L - F}{}$, then $u$ has two copies $u_0$ and $u_1$.
		$u_0$ starts with input $0$ while $u_1$ starts with input $1$.
		We denote the two sets corresponding to the two copies of nodes in $\inneighborhood{C}{R - F}{} - \inneighborhood{C}{L - F}{}$ by $\del[1]{\inneighborhood{C}{R - F}{} - \inneighborhood{C}{L - F}{}}_0$ and $\del[1]{\inneighborhood{C}{R - F}{} - \inneighborhood{C}{L - F}{}}_1$.
		
		\item
		If $u \in C - \inneighborhood{C}{R-F}{} - \inneighborhood{C}{L-F}{}$, then $u$ has three copies $u_0$, $u_1$, and $u'_1$.
		$u_0$ starts with input $0$ while $u_1$ and $u'_1$ start with input $1$.
		$C_0$, $C_1$, and $C'_1$ denote the three sets corresponding to the three copies of nodes in $C - \inneighborhood{C}{R-F}{} - \inneighborhood{C}{L-F}{}$.
	\end{enumerate}
	
	We use the following convention to depict edges of $\mathcal{G}$ in the figures.
	Consider any two distinct sets of nodes $A$ and $B$ of graph $G$, and two nodes $u \in A$ and $v \in B$.
	Let $A_i$ be one copy of $A$ and $B_j$ a copy of $B$ in $\mathcal{G}$, with $u_i \in A_i$ the corresponding copy of $u$ and $v_j \in B_j$ the corresponding copy of $v$.
	There are the following cases to consider.

	\begin{enumerate}[1)]
		\item
		There is an undirected edge between $A_i$ and $B_j$.
		Then $\mathcal{G}$ has following edges between $u_i$ and $v_j$:
		\begin{itemize}
			\item If $(u, v)$ is an edge in $G$, then $(u_i, v_j)$ is an edge in $\mathcal{G}$.
			\item If $(v, u)$ is an edge in $G$, then $(v_j, u_i)$ is an edge in $\mathcal{G}$.
		\end{itemize}
		
		\item
		There is a black directed edge from $A_i$ to $B_j$ (the case where the edge is from $B_j$ to $A_i$ follows similarly).
		Then, in $G$, the only edges between $A$ and $B$ are from $A$ to $B$, and there are no edges from $B$ to $A$.
		In this case, if $(u, v)$ is an edge in $G$, then $(u_i, v_j)$ is an edge in $\mathcal{G}$.
		
		\item
		There is a red directed edge from $A_i$ to $B_j$ (the case where the edge is from $B_j$ to $A_i$ follows similarly).
		Then, in $G$, the edges between $A$ and $B$ could be in either direction.
		$\mathcal{G}$ has following edges between $u_i$ and $v_j$:
		\begin{itemize}
			\item If $(u, v)$ is an edge in $G$, then $(u_i, v_j)$ is an edge in $\mathcal{G}$.
			\item $(v_j, u_i)$ is \underline{not} an edge in $\mathcal{G}$, even if $(v, u)$ is an edge in $G$.
		\end{itemize}
		
		\item
		There is no edge between $A_i$ and $B_j$.
		Then there are no edges between $u_i$ and $v_j$ in $\mathcal{G}$, even if there is an edge between the two nodes in $G$.
	\end{enumerate}
	
	All messages are transmitted via local broadcast in $\mathcal{G}$.
	In the construction of $\mathcal{G}$, we have carefully maintained the following property: if $(u, v)$ is an edge in the original graph $G$, then each copy of $v$ has exactly one copy of $u$ as an in-neighbor in $\mathcal{G}$.
	Note, however, that there could be multiple copies of $v$ (or none) that are out-neighbors of a copy of $u$.
	This allows us to create an algorithm for $\mathcal{G}$ corresponding to $\mathcal{A}$ as follows.
	For each node $u$ in $G$, each copy of $u$ in $\mathcal{G}$ runs $\mathcal{A}_u$.
	
	Now, consider an execution $\mathcal{E}$ of the above algorithm on $\mathcal{G}$ as follows.
	The inputs to each node were described earlier when listing the different copies of nodes.
	Observe that it is not guaranteed that nodes in $\mathcal{G}$ will reach an agreement or even if the algorithm will terminate.
	We will show that the algorithm indeed terminates but nodes do not reach agreement, which will be useful in deriving the desired contradiction.
	We use $\mathcal{E}$ to describe three executions $E_1$, $E_2$, and $E_3$ of $\mathcal{A}$ on the original graph $G$ as follows.
	
	\begin{enumerate}[$E_1$:]
		\item
		$\inneighborhood{R \cup C}{L - F}{}$ is the set of faulty nodes.
		In each round, a faulty node broadcasts the same messages as the corresponding copy of the node in $\mathcal{G}$ in execution $\mathcal{E}$ in the same round.
		All non-faulty nodes have input $0$.
		The behavior of non-faulty nodes is also modeled by the corresponding copies of nodes in $\mathcal{E}$.
		Figures \ref{figure necessity execution 1}-\ref{figure necessity C execution 1} depict the copies of nodes in $\mathcal{G}$ that model $E_1$, with faulty nodes drawn in red.
		Since $\mathcal{A}$ solves Byzantine consensus on $G$, so non-faulty nodes decide on output $0$ by validity in finite time.
		
		\item
		$\inneighborhood{L \cup C}{R - F}{}$ is the set of faulty nodes.
		In each round, a faulty node broadcasts the same messages as the corresponding copy of the node in $\mathcal{G}$ in execution $\mathcal{E}$ in the same round.
		All non-faulty nodes have input $1$.
		The behavior of non-faulty nodes is also modeled by the corresponding copies of nodes in $\mathcal{E}$.
		Figures \ref{figure necessity execution 2}-\ref{figure necessity C execution 2} depict the copies of nodes in $\mathcal{G}$ that model $E_2$, with faulty nodes drawn in red.
		Since $\mathcal{A}$ solves Byzantine consensus on $G$, so non-faulty nodes decide on output $1$ by validity in finite time.
		
		\item
		$F \cap (L \cup R)$ is the set of faulty nodes.
		In each round, a faulty node broadcasts the same messages as the corresponding copy of the node in $\mathcal{G}$ in execution $\mathcal{E}$ in the same round.
		Non-faulty nodes in $L$ have input $0$ while all other non-faulty nodes have input $1$.
		The behavior of non-faulty nodes is also modeled by the corresponding copies of nodes in $\mathcal{E}$.
		Figures \ref{figure necessity execution 3}-\ref{figure necessity C execution 3} depict the copies of nodes in $\mathcal{G}$ that model $E_3$, with faulty nodes drawn in red.
		The output of non-faulty nodes in $E_3$ will be described later.
	\end{enumerate}

	Since all non-faulty nodes in $L - F$ output $0$ in $E_1$, so nodes in $L_0 - F_0$ output $0$ in $\mathcal{E}$.
	Similarly, all non-faulty nodes in $R - F$ output $1$ in $E_2$ and so nodes in $R_0 - F_0$ output $1$ in $\mathcal{E}$.
	It follows that in $E_3$, nodes in $L - F$, as modeled by $L_0 - F_0$ in $\mathcal{E}$, output $0$ while nodes in $R - F$, as modeled by $R_0 - F_0$ in $\mathcal{E}$, output $1$.
	Since both $L - F$ and $R - F$ are non-empty by construction, we have that agreement is violated in $E_3$, a contradiction.
\end{proof} 		\begin{figure}[htb]
	\centering
	\begin{tikzpicture}[scale=0.75, every node/.style={scale=0.75}]
	\node[draw, circle, minimum size=1cm, gray] at (0, 0) (L_1 cap F) {$L_1 \cap F_1$};
	\node[draw, circle, minimum size=1cm, gray] at (8, 0) (R_1 cap F) {$R_1 \cap F_1$};
	\node[draw, circle, minimum size=1cm, gray] at (0, -3) (L_1) {$L_1 - F_1$};
	\node[draw, circle, minimum size=1cm, gray] at (8, -3) (R_1) {$R_1 - F_1$};
	
	\node[draw, ellipse, minimum height=2cm] at (0, -6) (NL cap FR) {$\inneighborhood{L \cap F}{R-F}{}$};
	\node[draw, ellipse, minimum height=2cm, red] at (8, -6) (NR cap FL) {$\inneighborhood{R \cap F}{L-F}{}$};
	\node[draw, ellipse, minimum height=2cm] at (0, -9) (NL minus FR) {$\inneighborhood{L - F}{R-F}{}$};
	\node[draw, ellipse, minimum height=2cm, red] at (8, -9) (NR minus FL) {$\inneighborhood{R - F}{L-F}{}$};
	
	\node[draw, circle, minimum size=1cm] at (0, -12) (L_0) {$L_0 - F_0$};
	\node[draw, circle, minimum size=1cm] at (8, -12) (R_0) {$R_0 - F_0$};
	\node[draw, circle, minimum size=1cm] at (0, -15) (L_0 cap F) {$L_0 \cap F_0$};
	\node[draw, circle, minimum size=1cm] at (8, -15) (R_0 cap F) {$R_0 \cap F_0$};

	\draw[-] (NR cap FL) to (NR minus FL);
	\draw[-] (NL cap FR) to (NL minus FR);
	\draw[-] (NL cap FR) to (NR cap FL);
	\draw[-] (NL minus FR) to (NR minus FL);
	\draw[-] (NL cap FR) to (NR minus FL);
	\draw[-] (NR cap FL) to (NL minus FR);
	
	\draw[-{Latex[width=3mm,length=3mm]},out=0, in=0,red] (NR cap FL) to (R_0 cap F);
	\draw[-{Latex[width=3mm,length=3mm]},out=0, in=0,red] (NR minus FL) to (R_0 cap F);
	\draw[-{Latex[width=3mm,length=3mm]},out=0, in=0,red] (NR cap FL) to (R_0);
	\draw[-{Latex[width=3mm,length=3mm]},out=0, in=0,red] (NR minus FL) to (R_0);
	
	\draw[-,out=180, in=180] (NL cap FR) to (L_0 cap F);
	\draw[-,out=180, in=180] (NL minus FR) to (L_0 cap F);
	\draw[-,out=180, in=180] (NL cap FR) to (L_0);
	\draw[-,out=180, in=180] (NL minus FR) to (L_0);
	
	\draw[-] (L_0 cap F) to (L_0);
	\draw[-] (R_0 cap F) to (R_0);
	
	\draw[-] (L_0 cap F) to (R_0 cap F);
	\draw[-{Latex[width=3mm,length=3mm]},red] (NR cap FL) to (L_0 cap F);
	\draw[-{Latex[width=3mm,length=3mm]}] (NR minus FL) to (L_0 cap F);
	\draw[-{Latex[width=3mm,length=3mm]}] (R_0) to (L_0 cap F);
	
	\draw[-] (NL cap FR) to (R_0 cap F);
	\draw[-{Latex[width=3mm,length=3mm]}] (NL minus FR) to (R_0 cap F);
	\draw[-{Latex[width=3mm,length=3mm]}] (L_0) to (R_0 cap F);
	
	\draw[-{Latex[width=3mm,length=3mm]},red] (NR cap FL) to (L_0);
	\draw[-{Latex[width=3mm,length=3mm]}] (NR minus FL) to (L_0);
	
	\draw[-] (NL cap FR) to (R_0);
	\draw[-{Latex[width=3mm,length=3mm]}] (NL minus FR) to (R_0);
	\end{tikzpicture}
	\caption{Execution $E_1$ as modeled by network $\mathcal{G}$: edges between copies of nodes in $L$ and $R$.}
	\label{figure necessity execution 1}
\end{figure}
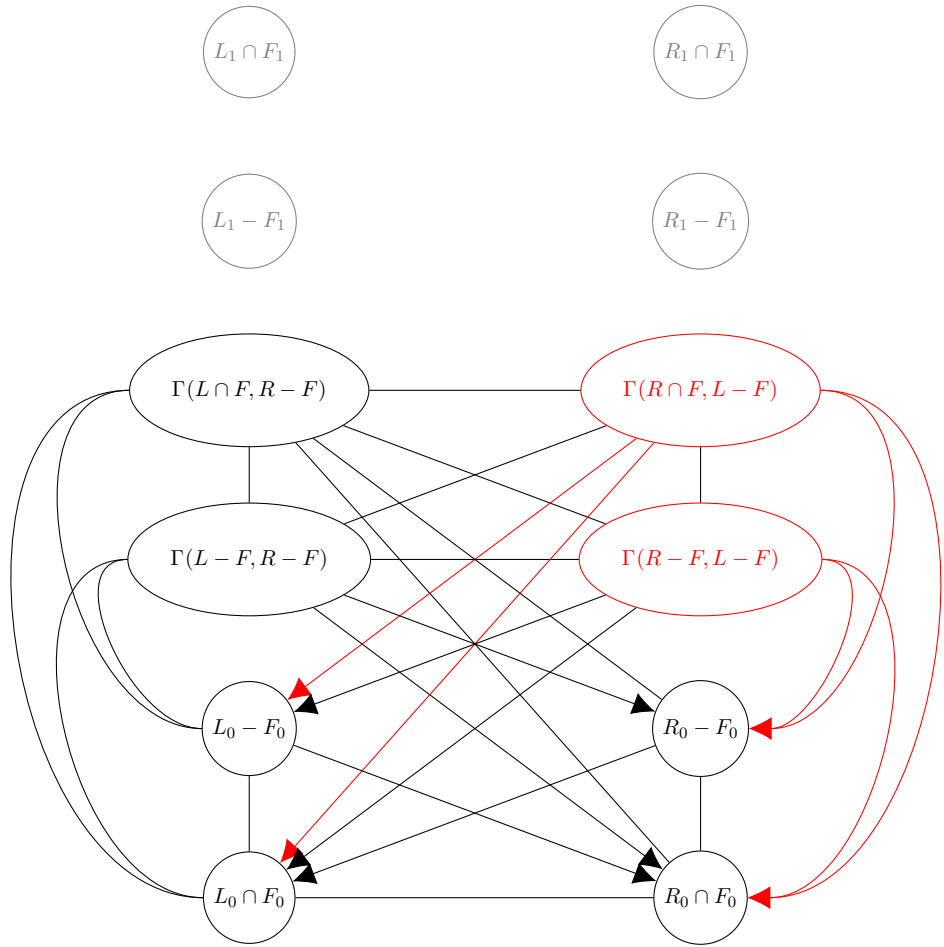

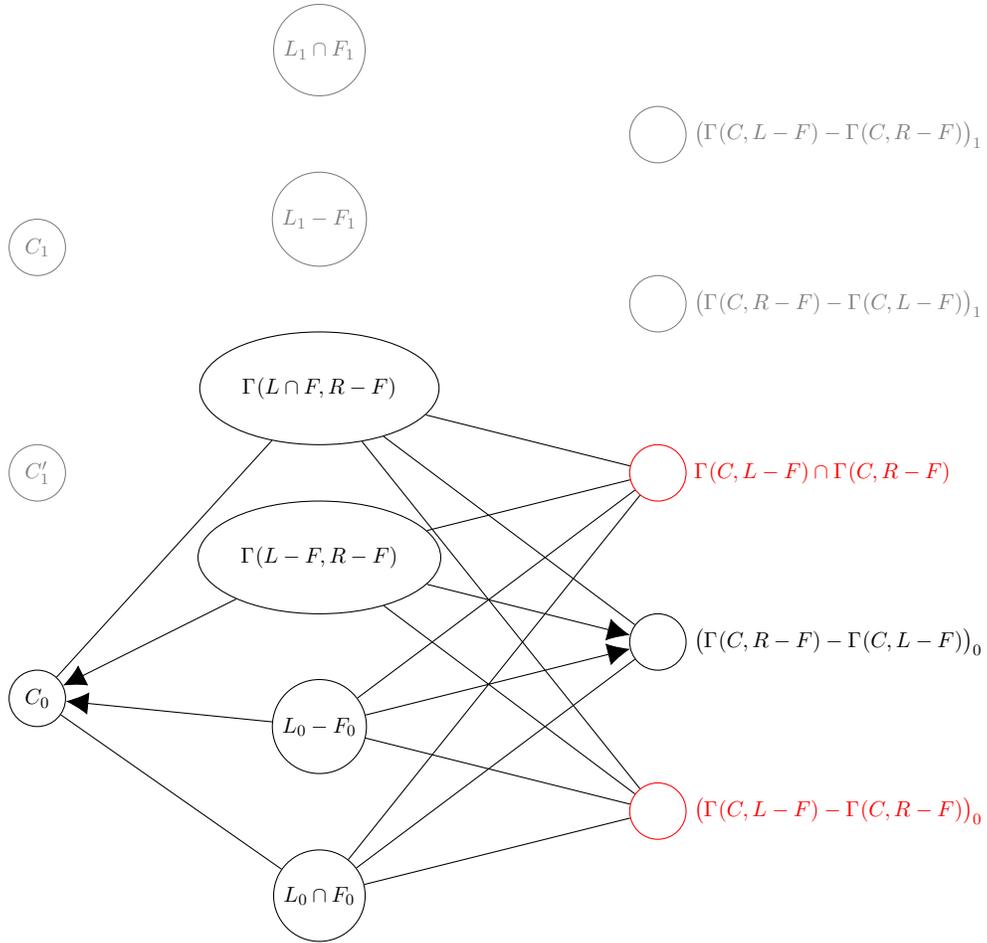
\begin{figure}[htb]
	\centering
	\begin{tikzpicture}[scale=0.75, every node/.style={scale=0.75}]
	\node[draw, circle, minimum size=1cm, gray] at (0, 0) (L_1 cap F) {$L_1 \cap F_1$};
	\node[draw, circle, minimum size=1cm, gray] at (0, -3) (L_1) {$L_1 - F_1$};
	\node[draw, ellipse, minimum height=2cm] at (0, -6) (NL cap FR) {$\inneighborhood{L \cap F}{R-F}{}$};
	\node[draw, ellipse, minimum height=2cm] at (0, -9) (NL minus FR) {$\inneighborhood{L - F}{R-F}{}$};
	\node[draw, circle, minimum size=1cm] at (0, -12) (L_0) {$L_0 - F_0$};
	\node[draw, circle, minimum size=1cm] at (0, -15) (L_0 cap F) {$L_0 \cap F_0$};
	
	\node[draw, circle, minimum size=1cm, label={[gray]right:{$\del{\inneighborhood{C}{R-F}{} - \inneighborhood{C}{L - F}{}}_1$}}, gray] at (6, -4.5) (NCR_1) {};
	\node[draw, circle, minimum size=1cm, label={[gray]right:{$\del{\inneighborhood{C}{L - F}{} - \inneighborhood{C}{R-F}{}}_1$}}, gray] at (6, -1.5) (NCL_1) {};
	\node[draw, circle, minimum size=1cm, label={[red]right:{$\inneighborhood{C}{L - F}{} \cap \inneighborhood{C}{R-F}{}$}}, red] at (6, -7.5) (NCLR) {};
	\node[draw, circle, minimum size=1cm, label={[red]right:{$\del{\inneighborhood{C}{L - F}{} - \inneighborhood{C}{R-F}{}}_0$}}, red] at (6, -13.5) (NCL_0) {};
	\node[draw, circle, minimum size=1cm, label=right:{{$\del{\inneighborhood{C}{R-F}{} - \inneighborhood{C}{L - F}{}}_0$}}] at (6, -10.5) (NCR_0) {};
	
	\node[draw, circle, minimum size=1cm] at (-5, -11.5) (C_0) {$C_0$};
	\node[draw, circle, minimum size=1cm, gray] at (-5, -7.5) (C'_1) {$C'_1$};
	\node[draw, circle, minimum size=1cm, gray] at (-5, -3.5) (C_1) {$C_1$};
	
	\draw[-] (L_0 cap F) to (C_0);
	\draw[-{Latex[width=3mm,length=3mm]}] (L_0) to (C_0);
	\draw[-] (NL cap FR) to (C_0);
	\draw[-{Latex[width=3mm,length=3mm]}] (NL minus FR) to (C_0);
	
	\draw[-] (L_0) to (NCL_0);
	\draw[-] (L_0) to (NCLR);
	\draw[-{Latex[width=3mm,length=3mm]}] (L_0) to (NCR_0);
	\draw[-] (L_0 cap F) to (NCL_0);
	\draw[-] (L_0 cap F) to (NCLR);
	\draw[-] (L_0 cap F) to (NCR_0);

	\draw[-] (NL cap FR) to (NCL_0);
	\draw[-] (NL cap FR) to (NCR_0);
	\draw[-] (NL cap FR) to (NCLR);
	
	\draw[-] (NL minus FR) to (NCL_0);
	\draw[-{Latex[width=3mm,length=3mm]}] (NL minus FR) to (NCR_0);
	\draw[-] (NL minus FR) to (NCLR);
	\end{tikzpicture}
	\caption{Execution $E_1$ as modeled by network $\mathcal{G}$: edges between copies of nodes in $L$ and $C$.}
	\label{figure necessity LC execution 1}
\end{figure}

\begin{figure}[htb]
	\centering
	\begin{tikzpicture}[scale=0.75, every node/.style={scale=0.75}]
	\node[draw, circle, minimum size=1cm, gray] at (0, 0) (R_1 cap F) {$R_1 \cap F_1$};
	\node[draw, circle, minimum size=1cm, gray] at (0, -3) (R_1) {$R_1 - F_1$};
	\node[draw, ellipse, minimum height=2cm, red] at (0, -6) (NR cap FL) {$\inneighborhood{R \cap F}{L - F}{}$};
	\node[draw, ellipse, minimum height=2cm, red] at (0, -9) (NR minus FL) {$\inneighborhood{R - F}{L - F}{}$};
	\node[draw, circle, minimum size=1cm] at (0, -12) (R_0) {$R_0 - F_0$};
	\node[draw, circle, minimum size=1cm] at (0, -15) (R_0 cap F) {$R_0 \cap F_0$};
	
	\node[draw, circle, minimum size=1cm, label={[gray]right:{$\del{\inneighborhood{C}{R - F}{} - \inneighborhood{C}{L - F}{}}_1$}}, gray] at (6, -4.5) (NCR_1) {};
	\node[draw, circle, minimum size=1cm, label={[gray]right:{$\del{\inneighborhood{C}{L - F}{} - \inneighborhood{C}{R - F}{}}_1$}}, gray] at (6, -1.5) (NCL_1) {};
	\node[draw, circle, minimum size=1cm, label={[red]right:{$\inneighborhood{C}{L - F}{} \cap \inneighborhood{C}{R - F}{}$}}, red] at (6, -7.5) (NCLR) {};
	\node[draw, circle, minimum size=1cm, label={[red]right:{$\del{\inneighborhood{C}{L - F}{} - \inneighborhood{C}{R - F}{}}_0$}}, red] at (6, -13.5) (NCL_0) {};
	\node[draw, circle, minimum size=1cm, label={right:{$\del{\inneighborhood{C}{R - F}{} - \inneighborhood{C}{L - F}{}}_0$}}] at (6, -10.5) (NCR_0) {};
	
	\node[draw, circle, minimum size=1cm] at (-5, -11.5) (C_0) {$C_0$};
	\node[draw, circle, minimum size=1cm, gray] at (-5, -7.5) (C'_1) {$C'_1$};
	\node[draw, circle, minimum size=1cm, gray] at (-5, -3.5) (C_1) {$C_1$};
	
	\draw[-] (R_0 cap F) to (C_0);
	\draw[-{Latex[width=3mm,length=3mm]}] (R_0) to (C_0);
	\draw[-{Latex[width=3mm,length=3mm]},red] (NR cap FL) to (C_0);
	\draw[-{Latex[width=3mm,length=3mm]}] (NR minus FL) to (C_0);
	
	\draw[-] (R_0) to (NCR_0);
	\draw[-{Latex[width=3mm,length=3mm]},red] (NCLR) to (R_0);
	\draw[-{Latex[width=3mm,length=3mm]},red] (NCLR) to (R_0 cap F);
	\draw[-] (R_0 cap F) to (NCR_0);
	
	\draw[-{Latex[width=3mm,length=3mm]},red] (NCL_0) to (R_0 cap F);
	
	\draw[-{Latex[width=3mm,length=3mm]},red] (NR cap FL) to (NCR_0);
	\draw[-{Latex[width=3mm,length=3mm]},red] (NR cap FL) to (NCL_0);
	\draw[-] (NR cap FL) to (NCLR);
	
	\draw[-{Latex[width=3mm,length=3mm]},red] (NR minus FL) to (NCR_0);
	\draw[-{Latex[width=3mm,length=3mm]}] (NR minus FL) to (NCL_0);
	\draw[-] (NR minus FL) to (NCLR);
	\end{tikzpicture}
	\caption{Execution $E_1$ as modeled by network $\mathcal{G}$: edges between copies of nodes in $R$ and $C$.}
	\label{figure necessity RC execution 1}
\end{figure}
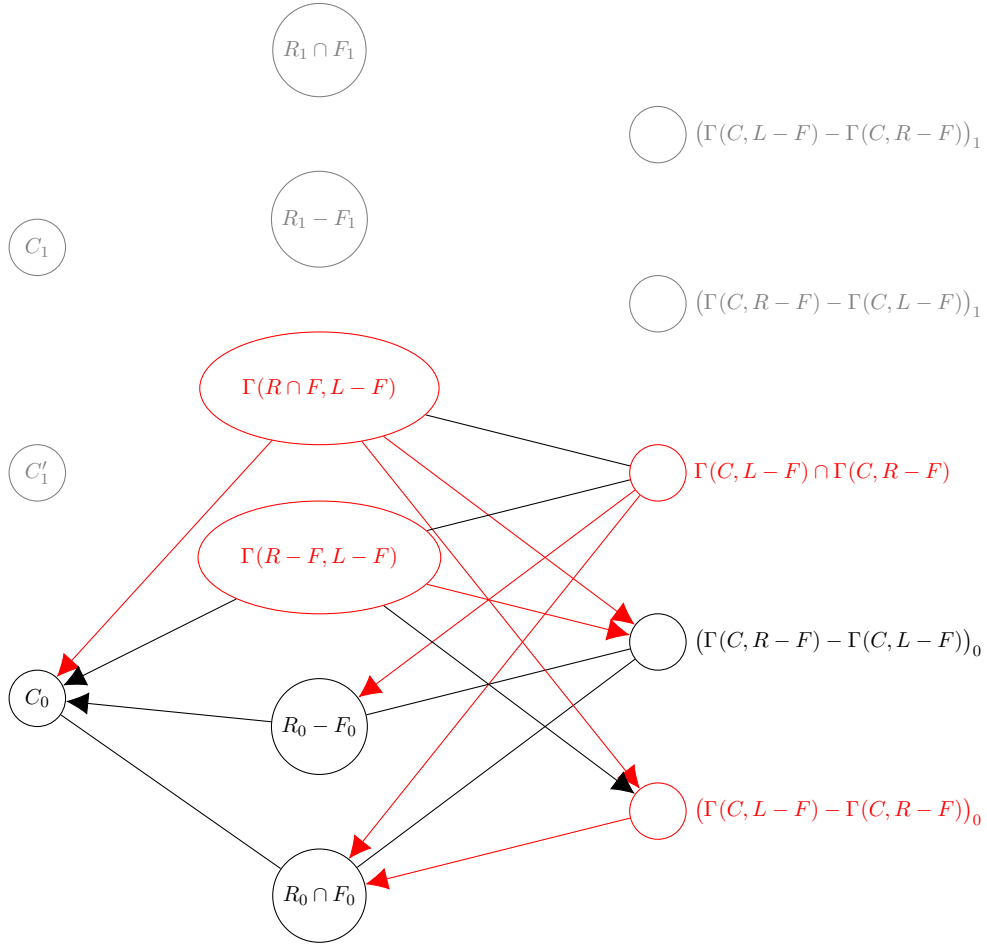

\begin{figure}[htb]
	\centering
	\begin{tikzpicture}[scale=0.85, every node/.style={scale=0.85}]
	\node[draw, circle, minimum size=1cm] at (135:4) (C_0) {$C_0$};
	\node[draw, circle, minimum size=1cm, gray] at (180:4) (C'_1) {$C'_1$};
	\node[draw, circle, minimum size=1cm, gray] at (225:4) (C_1) {$C_1$};
	\node[draw, circle, minimum size=1cm, label={[red]below:{$\del{\inneighborhood{C}{L - F}{} - \inneighborhood{C}{R - F}{}}_0$}}, red] at (270:4) (NCL_0) {};
	\node[draw, circle, minimum size=1cm, label={below right:{$\del{\inneighborhood{C}{R - F}{} - \inneighborhood{C}{L - F}{}}_0$}}] at (315:4) (NCR_0) {};
	\node[draw, circle, minimum size=1cm, label={[red]right:{$\inneighborhood{C}{L - F}{} \cap \inneighborhood{C}{R - F}{}$}}, red] at (360:4) (NCLR) {};
	\node[draw, circle, minimum size=1cm, label={[gray]above right:{$\del{\inneighborhood{C}{L - F}{} - \inneighborhood{C}{R - F}{}}_1$}}, gray] at (45:4) (NCL_1) {};
	\node[draw, circle, minimum size=1cm, label={[gray]above:{$\del{\inneighborhood{C}{R - F}{} - \inneighborhood{C}{L - F}{}}_1$}}, gray] at (90:4) (NCR_1) {};
	
	\draw[-{Latex[width=3mm,length=3mm]},red] (NCLR) to (C_0);
	\draw[-] (NCLR) to (NCL_0);
	\draw[-{Latex[width=3mm,length=3mm]},red] (NCLR) to (NCR_0);
	
	\draw[-{Latex[width=3mm,length=3mm]},red] (NCL_0) to (C_0);
	\draw[-{Latex[width=3mm,length=3mm]},red] (NCL_0) to (NCR_0);

	\draw[-] (NCR_0) to (C_0);
	\end{tikzpicture}
	\caption{Execution $E_1$ as modeled by network $\mathcal{G}$: edges between copies of nodes in $C$.}
	\label{figure necessity C execution 1}
\end{figure}
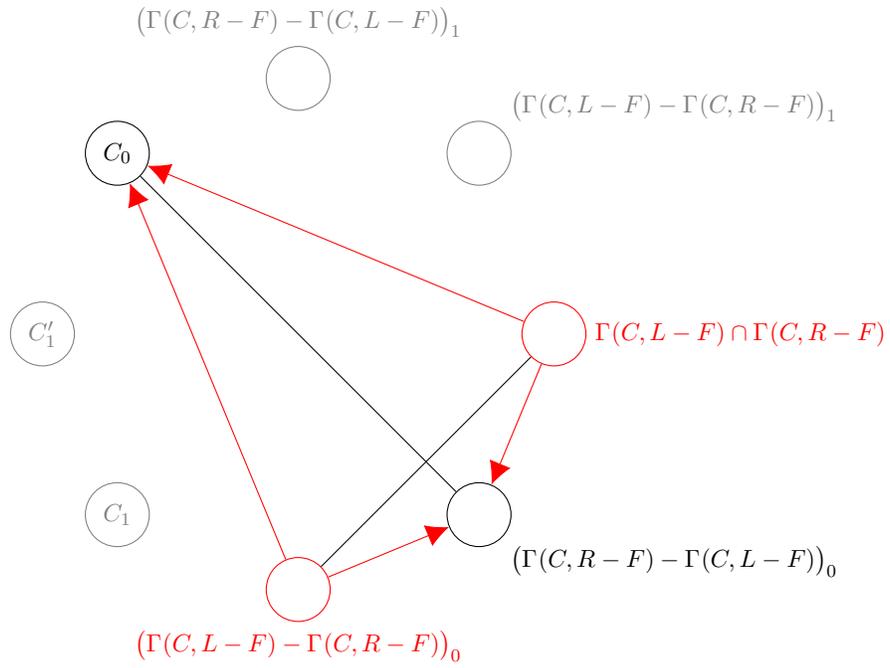 		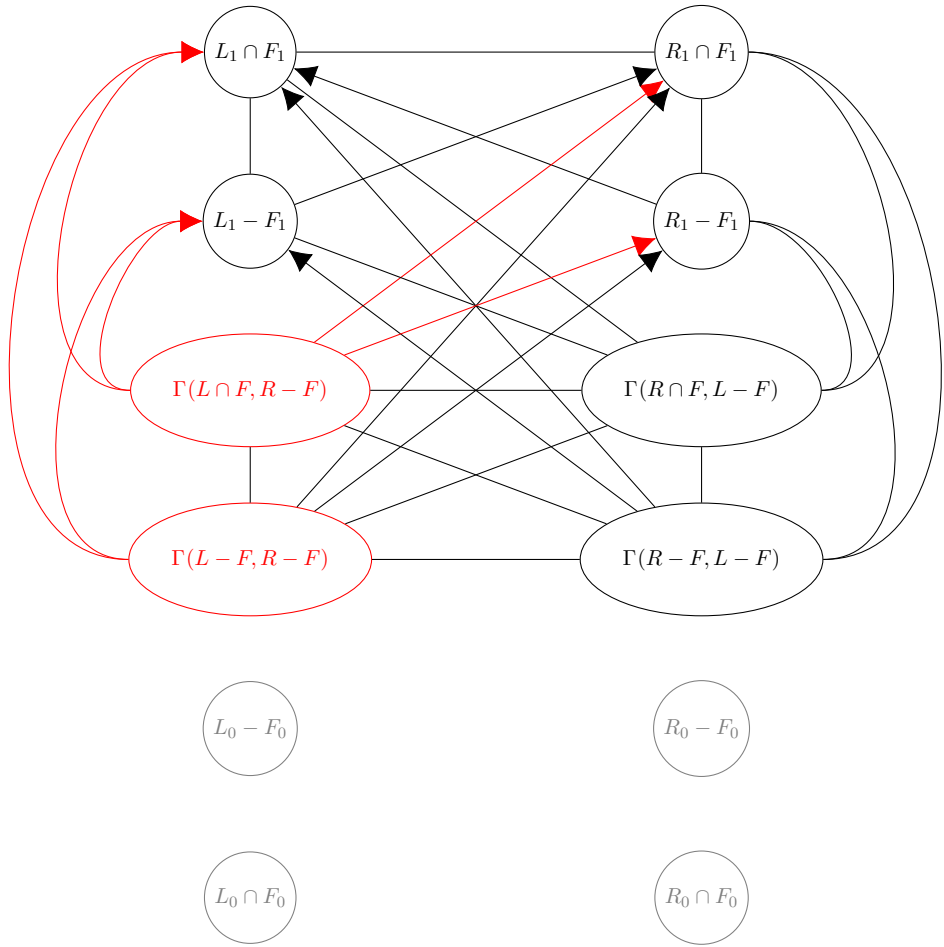
\begin{figure}[htb]
	\centering
	\begin{tikzpicture}[scale=0.75, every node/.style={scale=0.75}]
	\node[draw, circle, minimum size=1cm] at (0, 0) (L_1 cap F) {$L_1 \cap F_1$};
	\node[draw, circle, minimum size=1cm] at (8, 0) (R_1 cap F) {$R_1 \cap F_1$};
	\node[draw, circle, minimum size=1cm] at (0, -3) (L_1) {$L_1 - F_1$};
	\node[draw, circle, minimum size=1cm] at (8, -3) (R_1) {$R_1 - F_1$};
	
	\node[draw, ellipse, minimum height=2cm, red] at (0, -6) (NL cap FR) {$\inneighborhood{L \cap F}{R-F}{}$};
	\node[draw, ellipse, minimum height=2cm] at (8, -6) (NR cap FL) {$\inneighborhood{R \cap F}{L-F}{}$};
	\node[draw, ellipse, minimum height=2cm, red] at (0, -9) (NL minus FR) {$\inneighborhood{L - F}{R-F}{}$};
	\node[draw, ellipse, minimum height=2cm] at (8, -9) (NR minus FL) {$\inneighborhood{R - F}{L-F}{}$};
	
	\node[draw, circle, minimum size=1cm, gray] at (0, -12) (L_0) {$L_0 - F_0$};
	\node[draw, circle, minimum size=1cm, gray] at (8, -12) (R_0) {$R_0 - F_0$};
	\node[draw, circle, minimum size=1cm, gray] at (0, -15) (L_0 cap F) {$L_0 \cap F_0$};
	\node[draw, circle, minimum size=1cm, gray] at (8, -15) (R_0 cap F) {$R_0 \cap F_0$};
	
		\draw[-{Latex[width=3mm,length=3mm]},out=180, in=180,red] (NL cap FR) to (L_1 cap F);
		\draw[-{Latex[width=3mm,length=3mm]},out=180, in=180,red] (NL minus FR) to (L_1 cap F);
		\draw[-{Latex[width=3mm,length=3mm]},out=180, in=180,red] (NL cap FR) to (L_1);
		\draw[-{Latex[width=3mm,length=3mm]},out=180, in=180,red] (NL minus FR) to (L_1);
	
		\draw[-,out=0, in=0] (NR cap FL) to (R_1 cap F);
		\draw[-,out=0, in=0] (NR minus FL) to (R_1 cap F);
		\draw[-,out=0, in=0] (NR cap FL) to (R_1);
		\draw[-,out=0, in=0] (NR minus FL) to (R_1);
	
		\draw[-] (L_1 cap F) to (L_1);
		\draw[-] (R_1 cap F) to (R_1);
	
		\draw[-] (L_1 cap F) to (R_1 cap F);
		\draw[-] (NR cap FL) to (L_1 cap F);
		\draw[-{Latex[width=3mm,length=3mm]}] (NR minus FL) to (L_1 cap F);
		\draw[-{Latex[width=3mm,length=3mm]}] (R_1) to (L_1 cap F);
	
		\draw[-{Latex[width=3mm,length=3mm]},red] (NL cap FR) to (R_1 cap F);
		\draw[-{Latex[width=3mm,length=3mm]}] (NL minus FR) to (R_1 cap F);
		\draw[-{Latex[width=3mm,length=3mm]}] (L_1) to (R_1 cap F);
	
		\draw[-] (NR cap FL) to (L_1);
		\draw[-{Latex[width=3mm,length=3mm]}] (NR minus FL) to (L_1);
	
		\draw[-{Latex[width=3mm,length=3mm]},red] (NL cap FR) to (R_1);
		\draw[-{Latex[width=3mm,length=3mm]}] (NL minus FR) to (R_1);
	
	\draw[-] (NR cap FL) to (NR minus FL);
	\draw[-] (NL cap FR) to (NL minus FR);
	\draw[-] (NL cap FR) to (NR cap FL);
	\draw[-] (NL minus FR) to (NR minus FL);
	\draw[-] (NL cap FR) to (NR minus FL);
	\draw[-] (NR cap FL) to (NL minus FR);
	
	\end{tikzpicture}
	\caption{Execution $E_2$ as modeled by network $\mathcal{G}$: edges between copies of nodes in $L$ and $R$.}
	\label{figure necessity execution 2}
\end{figure}

\begin{figure}[htb]
	\centering
	\begin{tikzpicture}[scale=0.75, every node/.style={scale=0.75}]
	\node[draw, circle, minimum size=1cm] at (0, 0) (L_1 cap F) {$L_1 \cap F_1$};
	\node[draw, circle, minimum size=1cm] at (0, -3) (L_1) {$L_1 - F_1$};
	\node[draw, ellipse, minimum height=2cm, red] at (0, -6) (NL cap FR) {$\inneighborhood{L \cap F}{R-F}{}$};
	\node[draw, ellipse, minimum height=2cm, red] at (0, -9) (NL minus FR) {$\inneighborhood{L - F}{R-F}{}$};
	\node[draw, circle, minimum size=1cm, gray] at (0, -12) (L_0) {$L_0 - F_0$};
	\node[draw, circle, minimum size=1cm, gray] at (0, -15) (L_0 cap F) {$L_0 \cap F_0$};
	
	\node[draw, circle, minimum size=1cm, label={[red]right:{$\del{\inneighborhood{C}{R-F}{} - \inneighborhood{C}{L - F}{}}_1$}}, red] at (6, -4.5) (NCR_1) {};
	\node[draw, circle, minimum size=1cm, label={right:{$\del{\inneighborhood{C}{L - F}{} - \inneighborhood{C}{R-F}{}}_1$}}] at (6, -1.5) (NCL_1) {};
	\node[draw, circle, minimum size=1cm, label={[red]right:{$\inneighborhood{C}{L - F}{} \cap \inneighborhood{C}{R-F}{}$}}, red] at (6, -7.5) (NCLR) {};
	\node[draw, circle, minimum size=1cm, label={[gray]right:{{$\del{\inneighborhood{C}{L - F}{} - \inneighborhood{C}{R-F}{}}_0$}}}, gray] at (6, -13.5) (NCL_0) {};
	\node[draw, circle, minimum size=1cm, label={[gray]right:{{$\del{\inneighborhood{C}{R-F}{} - \inneighborhood{C}{L - F}{}}_0$}}}, gray] at (6, -10.5) (NCR_0) {};
	
	\node[draw, circle, minimum size=1cm, gray] at (-5, -11.5) (C_0) {$C_0$};
	\node[draw, circle, minimum size=1cm, gray] at (-5, -7.5) (C'_1) {$C'_1$};
	\node[draw, circle, minimum size=1cm] at (-5, -3.5) (C_1) {$C_1$};
	
		\draw[-] (L_1 cap F) to (C_1);
		\draw[-{Latex[width=3mm,length=3mm]}] (L_1) to (C_1);
		\draw[-{Latex[width=3mm,length=3mm]},red] (NL cap FR) to (C_1);
		\draw[-{Latex[width=3mm,length=3mm]}] (NL minus FR) to (C_1);

		\draw[-] (L_1) to (NCL_1);
		\draw[-{Latex[width=3mm,length=3mm]},red] (NCLR) to (L_1);
		\draw[-] (L_1 cap F) to (NCL_1);
		\draw[-{Latex[width=3mm,length=3mm]},red] (NCR_1) to (L_1 cap F);
		\draw[-{Latex[width=3mm,length=3mm]},red] (NCLR) to (L_1 cap F);
	
		\draw[-{Latex[width=3mm,length=3mm]},red] (NL cap FR) to (NCL_1);
		\draw[-{Latex[width=3mm,length=3mm]},red] (NL cap FR) to (NCR_1);
	\draw[-] (NL cap FR) to (NCLR);
	
		\draw[-{Latex[width=3mm,length=3mm]},red] (NL minus FR) to (NCL_1);
		\draw[-{Latex[width=3mm,length=3mm]}] (NL minus FR) to (NCR_1);
	\draw[-] (NL minus FR) to (NCLR);
	\end{tikzpicture}
	\caption{Execution $E_2$ as modeled by network $\mathcal{G}$: edges between copies of nodes in $L$ and $C$.}
	\label{figure necessity LC execution 2}
\end{figure}
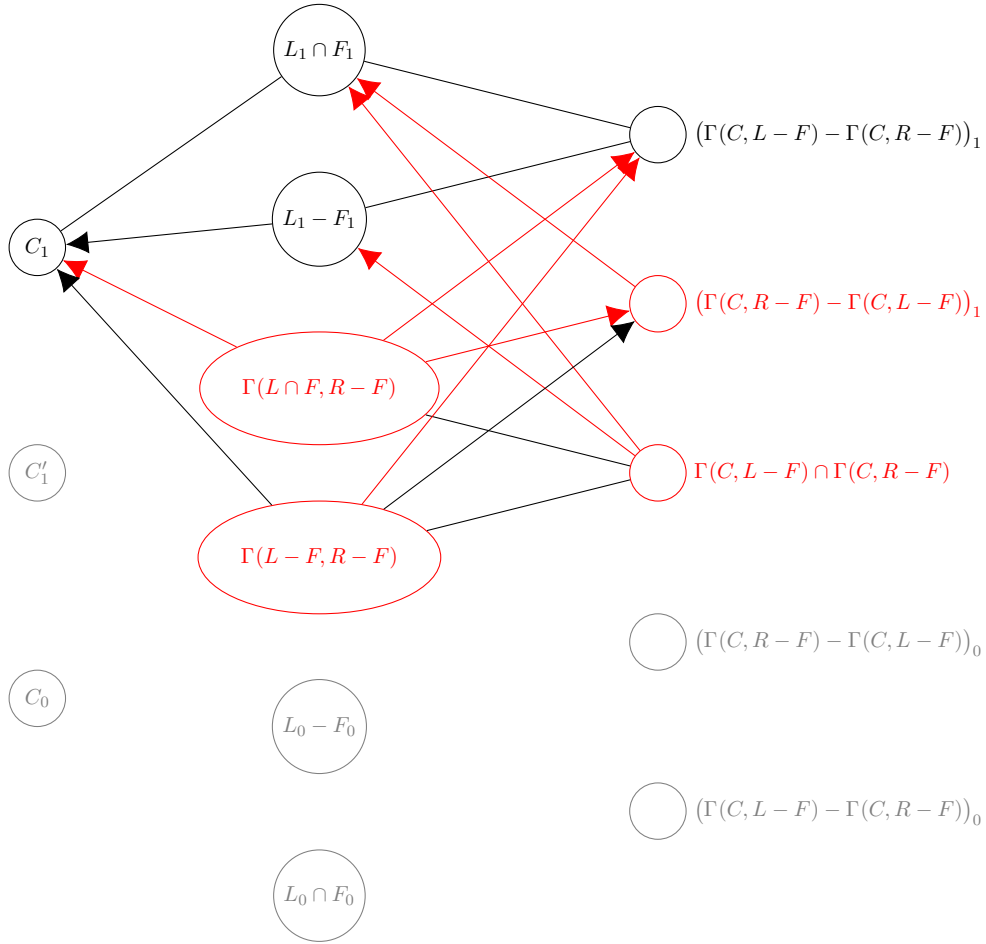

\begin{figure}[htb]
	\centering
	\begin{tikzpicture}[scale=0.75, every node/.style={scale=0.75}]
	\node[draw, circle, minimum size=1cm] at (0, 0) (R_1 cap F) {$R_1 \cap F_1$};
	\node[draw, circle, minimum size=1cm] at (0, -3) (R_1) {$R_1 - F_1$};
	\node[draw, ellipse, minimum height=2cm] at (0, -6) (NR cap FL) {$\inneighborhood{R \cap F}{L - F}{}$};
	\node[draw, ellipse, minimum height=2cm] at (0, -9) (NR minus FL) {$\inneighborhood{R - F}{L - F}{}$};
	\node[draw, circle, minimum size=1cm, gray] at (0, -12) (R_0) {$R_0 - F_0$};
	\node[draw, circle, minimum size=1cm, gray] at (0, -15) (R_0 cap F) {$R_0 \cap F_0$};
	
	\node[draw, circle, minimum size=1cm, label={[red]right:{$\del{\inneighborhood{C}{R - F}{} - \inneighborhood{C}{L - F}{}}_1$}}, red] at (6, -4.5) (NCR_1) {};
	\node[draw, circle, minimum size=1cm, label={right:{$\del{\inneighborhood{C}{L - F}{} - \inneighborhood{C}{R - F}{}}_1$}}] at (6, -1.5) (NCL_1) {};
	\node[draw, circle, minimum size=1cm, label={[red]right:{$\inneighborhood{C}{L - F}{} \cap \inneighborhood{C}{R - F}{}$}}, red] at (6, -7.5) (NCLR) {};
	\node[draw, circle, minimum size=1cm, label={[gray]right:{$\del{\inneighborhood{C}{L - F}{} - \inneighborhood{C}{R - F}{}}_0$}}, gray] at (6, -13.5) (NCL_0) {};
	\node[draw, circle, minimum size=1cm, label={[gray]right:{$\del{\inneighborhood{C}{R - F}{} - \inneighborhood{C}{L - F}{}}_0$}}, gray] at (6, -10.5) (NCR_0) {};
	
	\node[draw, circle, minimum size=1cm, gray] at (-5, -11.5) (C_0) {$C_0$};
	\node[draw, circle, minimum size=1cm, gray] at (-5, -7.5) (C'_1) {$C'_1$};
	\node[draw, circle, minimum size=1cm] at (-5, -3.5) (C_1) {$C_1$};
	
		\draw[-] (R_1 cap F) to (C_1);
		\draw[-{Latex[width=3mm,length=3mm]}] (R_1) to (C_1);
		\draw[-] (NR cap FL) to (C_1);
		\draw[-{Latex[width=3mm,length=3mm]}] (NR minus FL) to (C_1);

		\draw[-] (R_1) to (NCR_1);
		\draw[-] (R_1) to (NCLR);
		\draw[-{Latex[width=3mm,length=3mm]}] (R_1) to (NCL_1);
		\draw[-] (R_1 cap F) to (NCR_1);
		\draw[-] (R_1 cap F) to (NCLR);
		\draw[-] (NCL_1) to (L_1 cap F);
	
		\draw[-] (NR cap FL) to (NCR_1);
		\draw[-] (NR cap FL) to (NCL_1);
	\draw[-] (NR cap FL) to (NCLR);
	
		\draw[-] (NR minus FL) to (NCR_1);
		\draw[-{Latex[width=3mm,length=3mm]}] (NR minus FL) to (NCL_1);
	\draw[-] (NR minus FL) to (NCLR);
	\end{tikzpicture}
	\caption{Execution $E_2$ as modeled by network $\mathcal{G}$: edges between copies of nodes in $R$ and $C$.}
	\label{figure necessity RC execution 2}
\end{figure}
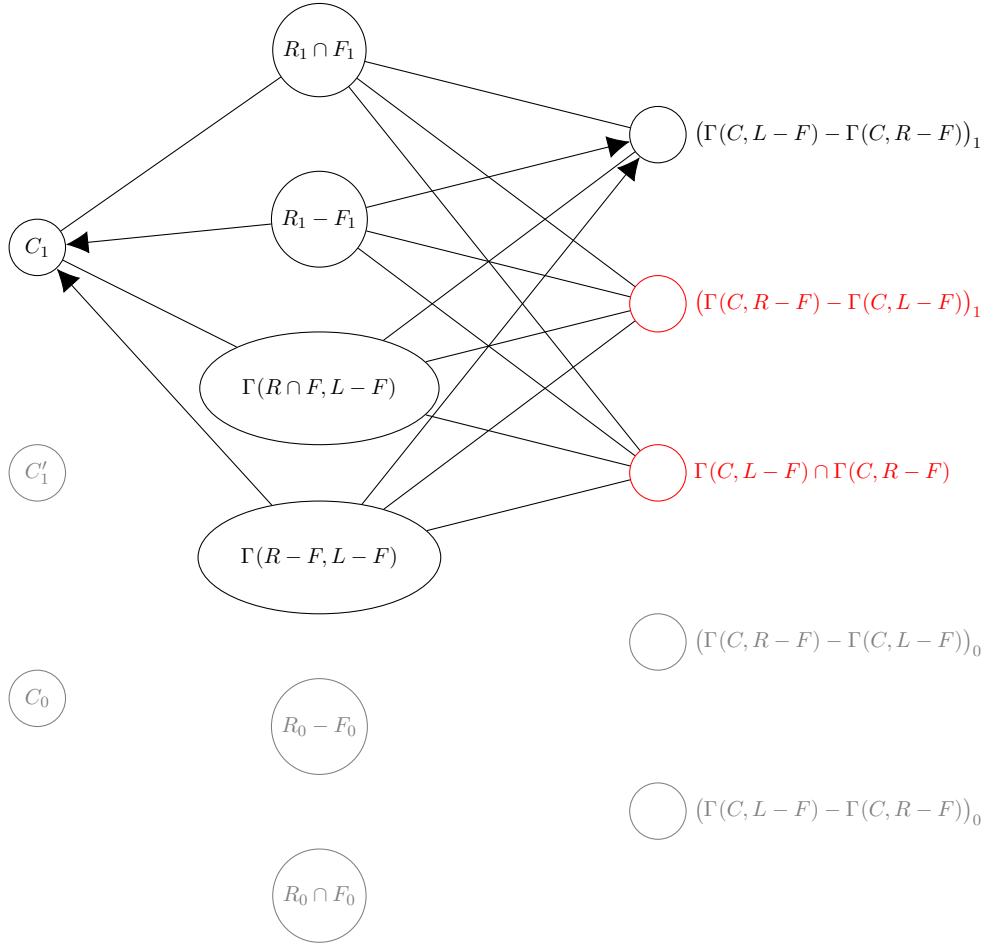

\begin{figure}[htb]
	\centering
	\begin{tikzpicture}[scale=0.85, every node/.style={scale=0.85}]
	\node[draw, circle, minimum size=1cm, gray] at (135:4) (C_0) {$C_0$};
	\node[draw, circle, minimum size=1cm, gray] at (180:4) (C'_1) {$C'_1$};
	\node[draw, circle, minimum size=1cm] at (225:4) (C_1) {$C_1$};
	\node[draw, circle, minimum size=1cm, label={[gray]below:{$\del{\inneighborhood{C}{L - F}{} - \inneighborhood{C}{R - F}{}}_0$}}, gray] at (270:4) (NCL_0) {};
	\node[draw, circle, minimum size=1cm, label={[gray]below right:{$\del{\inneighborhood{C}{R - F}{} - \inneighborhood{C}{L - F}{}}_0$}}, gray] at (315:4) (NCR_0) {};
	\node[draw, circle, minimum size=1cm, label={[red]right:{$\inneighborhood{C}{L - F}{} \cap \inneighborhood{C}{R - F}{}$}}, red] at (360:4) (NCLR) {};
	\node[draw, circle, minimum size=1cm, label={above right:{$\del{\inneighborhood{C}{L - F}{} - \inneighborhood{C}{R - F}{}}_1$}}] at (45:4) (NCL_1) {};
	\node[draw, circle, minimum size=1cm, label={[red]above:{$\del{\inneighborhood{C}{R - F}{} - \inneighborhood{C}{L - F}{}}_1$}}, red] at (90:4) (NCR_1) {};
	
		\draw[-{Latex[width=3mm,length=3mm]},red] (NCLR) to (C_1);
		\draw[-{Latex[width=3mm,length=3mm]},red] (NCLR) to (NCL_1);
		\draw[-] (NCLR) to (NCR_1);

		\draw[-{Latex[width=3mm,length=3mm]},red] (NCR_1) to (C_1);
		\draw[-{Latex[width=3mm,length=3mm]},red] (NCR_1) to (NCL_1);
	
		\draw[-] (NCL_1) to (C_1);
	\end{tikzpicture}
	\caption{Execution $E_2$ as modeled by network $\mathcal{G}$: edges between copies of nodes in $C$.}
	\label{figure necessity C execution 2}
\end{figure}
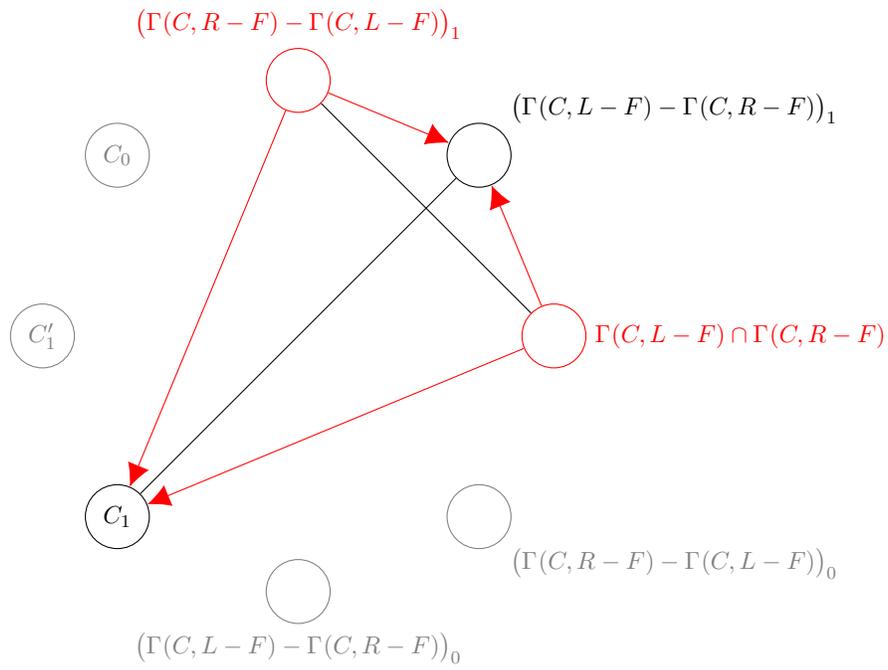 		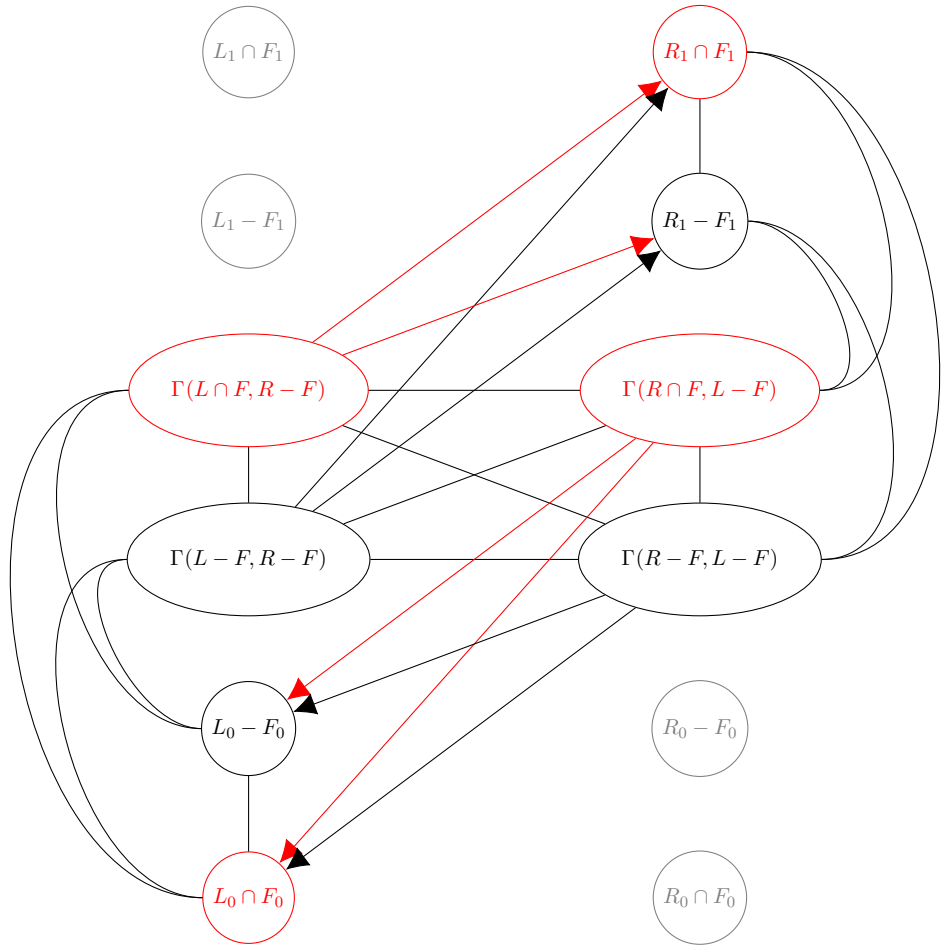
\begin{figure}[htb]
	\centering
	\begin{tikzpicture}[scale=0.75, every node/.style={scale=0.75}]
	\node[draw, circle, minimum size=1cm, gray] at (0, 0) (L_1 cap F) {$L_1 \cap F_1$};
	\node[draw, circle, minimum size=1cm, red] at (8, 0) (R_1 cap F) {$R_1 \cap F_1$};
	\node[draw, circle, minimum size=1cm, gray] at (0, -3) (L_1) {$L_1 - F_1$};
	\node[draw, circle, minimum size=1cm] at (8, -3) (R_1) {$R_1 - F_1$};
	
	\node[draw, ellipse, minimum height=2cm, red] at (0, -6) (NL cap FR) {$\inneighborhood{L \cap F}{R-F}{}$};
	\node[draw, ellipse, minimum height=2cm, red] at (8, -6) (NR cap FL) {$\inneighborhood{R \cap F}{L-F}{}$};
	\node[draw, ellipse, minimum height=2cm] at (0, -9) (NL minus FR) {$\inneighborhood{L - F}{R-F}{}$};
	\node[draw, ellipse, minimum height=2cm] at (8, -9) (NR minus FL) {$\inneighborhood{R - F}{L-F}{}$};
	
	\node[draw, circle, minimum size=1cm] at (0, -12) (L_0) {$L_0 - F_0$};
	\node[draw, circle, minimum size=1cm, gray] at (8, -12) (R_0) {$R_0 - F_0$};
	\node[draw, circle, minimum size=1cm, red] at (0, -15) (L_0 cap F) {$L_0 \cap F_0$};
	\node[draw, circle, minimum size=1cm, gray] at (8, -15) (R_0 cap F) {$R_0 \cap F_0$};

		\draw[-,out=0, in=0] (NR cap FL) to (R_1 cap F);
		\draw[-,out=0, in=0] (NR minus FL) to (R_1 cap F);
		\draw[-,out=0, in=0] (NR cap FL) to (R_1);
		\draw[-,out=0, in=0] (NR minus FL) to (R_1);
	
		\draw[-] (R_1 cap F) to (R_1);

		\draw[-{Latex[width=3mm,length=3mm]},red] (NL cap FR) to (R_1 cap F);
		\draw[-{Latex[width=3mm,length=3mm]}] (NL minus FR) to (R_1 cap F);

		\draw[-{Latex[width=3mm,length=3mm]},red] (NL cap FR) to (R_1);
		\draw[-{Latex[width=3mm,length=3mm]}] (NL minus FR) to (R_1);
	
	\draw[-] (NR cap FL) to (NR minus FL);
	\draw[-] (NL cap FR) to (NL minus FR);
	\draw[-] (NL cap FR) to (NR cap FL);
	\draw[-] (NL minus FR) to (NR minus FL);
	\draw[-] (NL cap FR) to (NR minus FL);
	\draw[-] (NR cap FL) to (NL minus FR);

	\draw[-,out=180, in=180] (NL cap FR) to (L_0 cap F);
	\draw[-,out=180, in=180] (NL minus FR) to (L_0 cap F);
	\draw[-,out=180, in=180] (NL cap FR) to (L_0);
	\draw[-,out=180, in=180] (NL minus FR) to (L_0);
	
	\draw[-] (L_0 cap F) to (L_0);
	
	\draw[-{Latex[width=3mm,length=3mm]},red] (NR cap FL) to (L_0 cap F);
	\draw[-{Latex[width=3mm,length=3mm]}] (NR minus FL) to (L_0 cap F);

	\draw[-{Latex[width=3mm,length=3mm]},red] (NR cap FL) to (L_0);
	\draw[-{Latex[width=3mm,length=3mm]}] (NR minus FL) to (L_0);
	
	\end{tikzpicture}
	\caption{Execution $E_3$ as modeled by network $\mathcal{G}$: edges between copies of nodes in $L$ and $R$.}
	\label{figure necessity execution 3}
\end{figure}

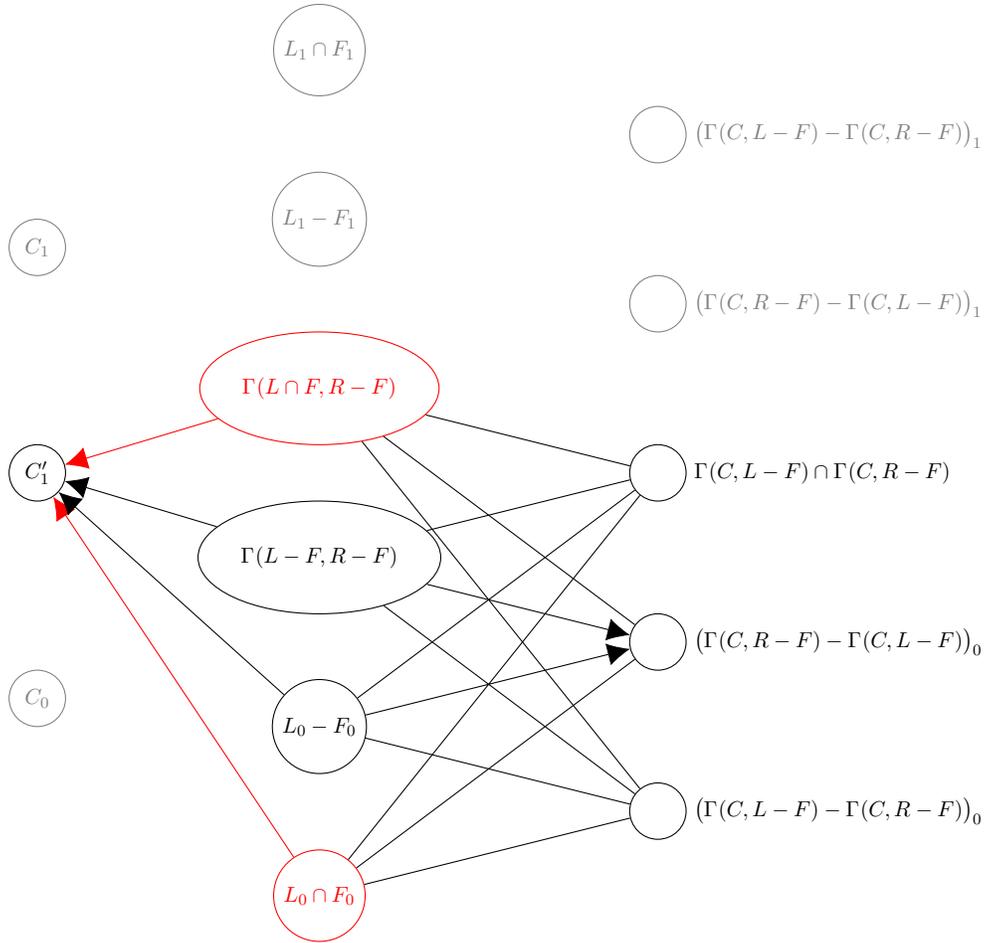
\begin{figure}[htb]
	\centering
	\begin{tikzpicture}[scale=0.75, every node/.style={scale=0.75}]
	\node[draw, circle, minimum size=1cm, gray] at (0, 0) (L_1 cap F) {$L_1 \cap F_1$};
	\node[draw, circle, minimum size=1cm, gray] at (0, -3) (L_1) {$L_1 - F_1$};
	\node[draw, ellipse, minimum height=2cm, red] at (0, -6) (NL cap FR) {$\inneighborhood{L \cap F}{R-F}{}$};
	\node[draw, ellipse, minimum height=2cm] at (0, -9) (NL minus FR) {$\inneighborhood{L - F}{R-F}{}$};
	\node[draw, circle, minimum size=1cm] at (0, -12) (L_0) {$L_0 - F_0$};
	\node[draw, circle, minimum size=1cm, red] at (0, -15) (L_0 cap F) {$L_0 \cap F_0$};
	
	\node[draw, circle, minimum size=1cm, label={[gray]right:{$\del{\inneighborhood{C}{R-F}{} - \inneighborhood{C}{L - F}{}}_1$}}, gray] at (6, -4.5) (NCR_1) {};
	\node[draw, circle, minimum size=1cm, label={[gray]right:{$\del{\inneighborhood{C}{L - F}{} - \inneighborhood{C}{R-F}{}}_1$}}, gray] at (6, -1.5) (NCL_1) {};
	\node[draw, circle, minimum size=1cm, label=right:{$\inneighborhood{C}{L - F}{} \cap \inneighborhood{C}{R-F}{}$}] at (6, -7.5) (NCLR) {};
	\node[draw, circle, minimum size=1cm, label=right:{{$\del{\inneighborhood{C}{L - F}{} - \inneighborhood{C}{R-F}{}}_0$}}] at (6, -13.5) (NCL_0) {};
	\node[draw, circle, minimum size=1cm, label=right:{{$\del{\inneighborhood{C}{R-F}{} - \inneighborhood{C}{L - F}{}}_0$}}] at (6, -10.5) (NCR_0) {};
	
	\node[draw, circle, minimum size=1cm, gray] at (-5, -11.5) (C_0) {$C_0$};
	\node[draw, circle, minimum size=1cm] at (-5, -7.5) (C'_1) {$C'_1$};
	\node[draw, circle, minimum size=1cm, gray] at (-5, -3.5) (C_1) {$C_1$};
	
		\draw[-{Latex[width=3mm,length=3mm]},red] (L_0 cap F) to (C'_1);
		\draw[-{Latex[width=3mm,length=3mm]}] (L_0) to (C'_1);
		\draw[-{Latex[width=3mm,length=3mm]},red] (NL cap FR) to (C'_1);
		\draw[-{Latex[width=3mm,length=3mm]}] (NL minus FR) to (C'_1);
	
	\draw[-] (L_0) to (NCL_0);
	\draw[-] (L_0) to (NCLR);
	\draw[-{Latex[width=3mm,length=3mm]}] (L_0) to (NCR_0);
	\draw[-] (L_0 cap F) to (NCL_0);
	\draw[-] (L_0 cap F) to (NCLR);
	\draw[-] (L_0 cap F) to (NCR_0);

	\draw[-] (NL cap FR) to (NCL_0);
	\draw[-] (NL cap FR) to (NCR_0);
	\draw[-] (NL cap FR) to (NCLR);
	
	\draw[-] (NL minus FR) to (NCL_0);
	\draw[-{Latex[width=3mm,length=3mm]}] (NL minus FR) to (NCR_0);
	\draw[-] (NL minus FR) to (NCLR);
	\end{tikzpicture}
	\caption{Execution $E_3$ as modeled by network $\mathcal{G}$: edges between copies of nodes in $L$ and $C$.}
	\label{figure necessity LC execution 3}
\end{figure}

\begin{figure}[htb]
	\centering
	\begin{tikzpicture}[scale=0.75, every node/.style={scale=0.75}]
	\node[draw, circle, minimum size=1cm, red] at (0, 0) (R_1 cap F) {$R_1 \cap F_1$};
	\node[draw, circle, minimum size=1cm] at (0, -3) (R_1) {$R_1 - F_1$};
	\node[draw, ellipse, minimum height=2cm, red] at (0, -6) (NR cap FL) {$\inneighborhood{R \cap F}{L - F}{}$};
	\node[draw, ellipse, minimum height=2cm] at (0, -9) (NR minus FL) {$\inneighborhood{R - F}{L - F}{}$};
	\node[draw, circle, minimum size=1cm, gray] at (0, -12) (R_0) {$R_0 - F_0$};
	\node[draw, circle, minimum size=1cm, gray] at (0, -15) (R_0 cap F) {$R_0 \cap F_0$};
	
	\node[draw, circle, minimum size=1cm, label={right:{$\del{\inneighborhood{C}{R - F}{} - \inneighborhood{C}{L - F}{}}_1$}}] at (6, -4.5) (NCR_1) {};
	\node[draw, circle, minimum size=1cm, label={right:{$\del{\inneighborhood{C}{L - F}{} - \inneighborhood{C}{R - F}{}}_1$}}] at (6, -1.5) (NCL_1) {};
	\node[draw, circle, minimum size=1cm, label=right:{$\inneighborhood{C}{L - F}{} \cap \inneighborhood{C}{R - F}{}$}] at (6, -7.5) (NCLR) {};
	\node[draw, circle, minimum size=1cm, label={[gray]right:{$\del{\inneighborhood{C}{L - F}{} - \inneighborhood{C}{R - F}{}}_0$}}, gray] at (6, -13.5) (NCL_0) {};
	\node[draw, circle, minimum size=1cm, label={[gray]right:{$\del{\inneighborhood{C}{R - F}{} - \inneighborhood{C}{L - F}{}}_0$}}, gray] at (6, -10.5) (NCR_0) {};
	
	\node[draw, circle, minimum size=1cm, gray] at (-5, -11.5) (C_0) {$C_0$};
	\node[draw, circle, minimum size=1cm] at (-5, -7.5) (C'_1) {$C'_1$};
	\node[draw, circle, minimum size=1cm, gray] at (-5, -3.5) (C_1) {$C_1$};
	
		\draw[-{Latex[width=3mm,length=3mm]},red] (R_1 cap F) to (C'_1);
		\draw[-{Latex[width=3mm,length=3mm]}] (R_1) to (C'_1);
		\draw[-{Latex[width=3mm,length=3mm]},red] (NR cap FL) to (C'_1);
		\draw[-{Latex[width=3mm,length=3mm]}] (NR minus FL) to (C'_1);

	\draw[-] (R_1) to (NCR_1);
	\draw[-] (R_1) to (NCLR);
	\draw[-{Latex[width=3mm,length=3mm]}] (R_1) to (NCL_1);
	\draw[-] (R_1 cap F) to (NCR_1);
	\draw[-] (R_1 cap F) to (NCLR);
	\draw[-] (NCL_1) to (L_1 cap F);
	
	\draw[-] (NR cap FL) to (NCR_1);
	\draw[-] (NR cap FL) to (NCL_1);
	\draw[-] (NR cap FL) to (NCLR);
	
	\draw[-] (NR minus FL) to (NCR_1);
	\draw[-{Latex[width=3mm,length=3mm]}] (NR minus FL) to (NCL_1);
	\draw[-] (NR minus FL) to (NCLR);
	\end{tikzpicture}
	\caption{Execution $E_3$ as modeled by network $\mathcal{G}$: edges between copies of nodes in $R$ and $C$.}
	\label{figure necessity RC execution 3}
\end{figure}
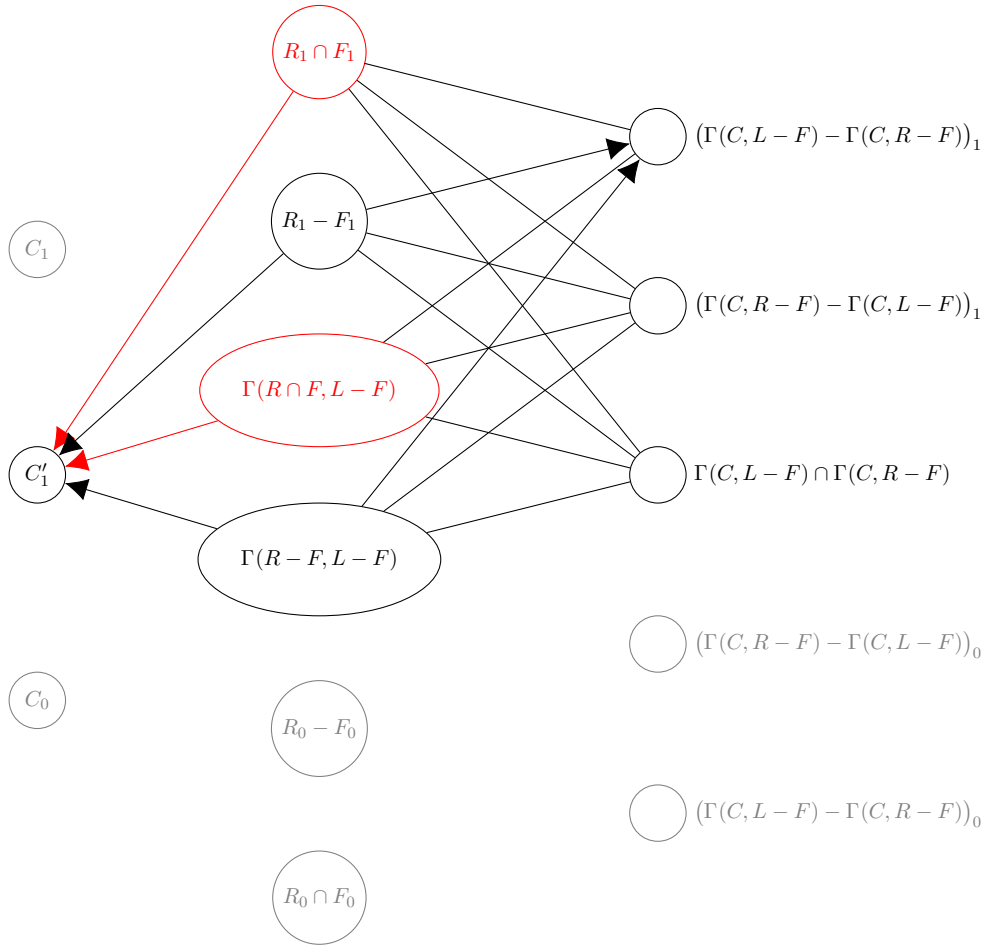

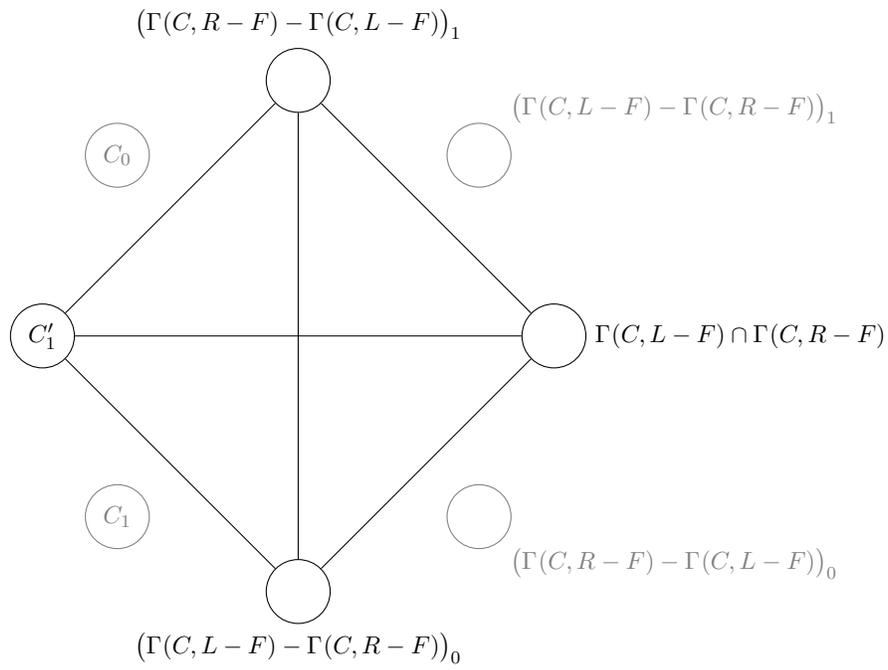
\begin{figure}[htb]
	\centering
	\begin{tikzpicture}[scale=0.85, every node/.style={scale=0.85}]
	\node[draw, circle, minimum size=1cm, gray] at (135:4) (C_0) {$C_0$};
	\node[draw, circle, minimum size=1cm] at (180:4) (C'_1) {$C'_1$};
	\node[draw, circle, minimum size=1cm, gray] at (225:4) (C_1) {$C_1$};
	\node[draw, circle, minimum size=1cm, label={below:{$\del{\inneighborhood{C}{L - F}{} - \inneighborhood{C}{R - F}{}}_0$}}] at (270:4) (NCL_0) {};
	\node[draw, circle, minimum size=1cm, label={[gray]below right:{$\del{\inneighborhood{C}{R - F}{} - \inneighborhood{C}{L - F}{}}_0$}}, gray] at (315:4) (NCR_0) {};
	\node[draw, circle, minimum size=1cm, label={right:{$\inneighborhood{C}{L - F}{} \cap \inneighborhood{C}{R - F}{}$}}] at (360:4) (NCLR) {};
	\node[draw, circle, minimum size=1cm, label={[gray]above right:{$\del{\inneighborhood{C}{L - F}{} - \inneighborhood{C}{R - F}{}}_1$}}, gray] at (45:4) (NCL_1) {};
	\node[draw, circle, minimum size=1cm, label={above:{$\del{\inneighborhood{C}{R - F}{} - \inneighborhood{C}{L - F}{}}_1$}}] at (90:4) (NCR_1) {};
	
		\draw[-] (NCLR) to (C'_1);
		\draw[-] (NCLR) to (NCL_0);
	\draw[-] (NCLR) to (NCR_1);
	
		\draw[-] (NCL_0) to (C'_1);
		\draw[-] (NCL_0) to (NCR_1);
	
		\draw[-] (NCR_1) to (C'_1);
	
	\end{tikzpicture}
	\caption{Execution $E_3$ as modeled by network $\mathcal{G}$: edges between copies of nodes in $C$.}
	\label{figure necessity C execution 3}
\end{figure} \end{document}